\documentclass[twocolumn,prx,superscriptaddress,longbibliography,nofootinbib]{revtex4-2}

\usepackage[dvips]{graphicx}
\usepackage{amsfonts,amscd,amsmath,amsthm}
\usepackage{enumerate}
\usepackage{epsfig}
\usepackage{subfigure}
\usepackage{xcolor}
\usepackage[colorlinks = true]{hyperref}
\usepackage{physics}
\usepackage{epstopdf}
\usepackage{framed}
\usepackage{multirow}
\usepackage{color}
\usepackage{comment}
\usepackage[ruled,vlined]{algorithm2e}
\usepackage[most]{tcolorbox}
\usepackage{caption}
\usepackage{subcaption}

\usepackage{tikz}
\usetikzlibrary{tikzmark, calc, fit, positioning}
\usetikzlibrary{shapes}
\usetikzlibrary{quantikz}

\newtheorem{theorem}{Theorem}
\newtheorem{lemma}{Lemma}

\newtheorem{definition}{Definition}
\newtheorem{proposition}{Proposition}

\newtheorem{remark}{Remark}

\newtcolorbox[auto counter]{mybox}[2][]{
	enhanced,
	breakable,
	colback=blue!5!white,
	colframe=blue!75!black,
	fonttitle=\bfseries,
	title=Box \thetcbcounter: #2,#1
}

\newtheorem*{theorem*}{Theorem}

\newcommand{\mbv}{\mathbf{v}}
\newcommand{\mbe}{\mathbf{e}}

\graphicspath{{figs}{figs/plots/GHZ}{figs/plots/toric}}

\begin{document}

\title{Variational LOCC-Assisted Quantum Circuits for Long-Range Entangled States}

\author{Yuxuan Yan}
\thanks{These authors contributed equally to this work.}
\affiliation{Center for Quantum Information, Institute for Interdisciplinary Information Sciences, Tsinghua University, Beijing 100084, China}

\author{Muzhou Ma}
\thanks{These authors contributed equally to this work.}
\affiliation{Department of Electronic Engineering, Tsinghua University, Beijing 100084, China}

\author{You Zhou}
\email{you\_zhou@fudan.edu.cn}
\affiliation{Key Laboratory for Information Science of Electromagnetic Waves (Ministry of Education), Fudan University, Shanghai 200433, China}
\author{Xiongfeng Ma}
\email{xma@tsinghua.edu.cn}
\affiliation{Center for Quantum Information, Institute for Interdisciplinary Information Sciences, Tsinghua University, Beijing 100084, China}

\begin{abstract}
Long-range entanglement is an important quantum resource, particularly for topological orders and quantum error correction. In reality, preparing long-range entangled states requires a deep unitary circuit, which poses significant experimental challenges. A promising avenue is offered by replacing some quantum resources with local operations and classical communication (LOCC). With these classical components, one can communicate outcomes of midcircuit measurements in distant subsystems, substantially reducing circuit depth in many important cases. However, to prepare general long-range entangled states, finding LOCC-assisted circuits of a short depth remains an open question. Here, to address this challenge, we propose a quantum-classical hybrid algorithm to find optimal LOCC protocols for preparing ground states of given Hamiltonians. In our algorithm, we introduce an efficient way to estimate parameter gradients and use such gradients for variational optimization. Theoretically, we establish the conditions for the absence of barren plateaus, ensuring trainability at a large system size. Numerically, the algorithm accurately solves the ground state of long-range entangled models, such as the perturbed Greenberger–Horne–Zeilinger state and surface code. Our results demonstrate the advantage of our method over conventional unitary variational circuits: the practical advantage in the accuracy of estimated ground-state energy and the theoretical advantage in creating long-range entanglement.
\end{abstract}

\maketitle

Long-range entanglement structures play an essential role in many quantum information processing scenarios. By definition, long-range entangled states require deep unitary circuits to be prepared from a product state; asymptotically, the depth requirement is unbounded when the system scales \cite{wen_topological_2013, chen_local_2010}. A typical example is the Greenberger–Horne–Zeilinger (GHZ) state \cite{greenberger_going_1989}, which finds important applications in quantum communication, cryptography, and computation. In addition, as a canonical example exhibiting topological orders, the surface code state is also long-range entangled and consequently serves as a resource for topological quantum memory and computation \cite{bravyi_quantum_1998, kitaev_fault-tolerant_2003}. More generally, quantum topological order and error-correcting codes essentially rely on long-range entanglement \cite{yi_complexity_2024}.

Important as they are, the preparation of long-range entangled states is severely challenged by their depth requirements \cite{aharonov_limitations_1996, muller-hermes_relative_2016,yan_limitations_2023}. Fortunately, a promising solution is found by introducing local operations and classical communication (LOCC). With the assistance of LOCC, circuits would include midcircuit measurements and feedforwards, where measurement results determine the subsequent quantum local operations. Note that LOCC-assisted circuits are also referred to by other terminologies, including circuits with midcircuit measurements, adaptive circuits, and dynamic circuits.
The essential role of LOCC in these circuits is to communicate information among distant subsystems and thereby create long-range correlations, which necessitates a significantly larger depth for circuits with only local unitary gates \cite{friedman_locality_2022}.
The introduction of LOCC brings great success in preparing surface code states in a constant depth \cite{raussendorf_long-range_2005, aguado_creation_2008, piroli_quantum_2021}, and it was later extended to other topologically ordered systems \cite{tantivasadakarn_long-range_2024, bravyi_adaptive_2022, tantivasadakarn_shortest_2023, tantivasadakarn_hierarchy_2023, li_symmetry-enriched_2023} and other important states for quantum information, such as the GHZ state \cite{piroli_quantum_2021, zhu_nishimoris_2023}, W states \cite{piroli_quantum_2021, buhrman_state_2024, piroli_approximating_2024}, and Dicke states \cite{buhrman_state_2024, piroli_approximating_2024}. Recently, progress has also been made in preparing tensor-network states \cite{lu_measurement_2022, smith_deterministic_2022, malz_preparation_2024, smith_constant-depth_2024, gunn_phases_2023}. These findings provide theoretical insights into the power of LOCC-assisted circuits in terms of depth reduction.

Despite the successes of previous important cases, the full potential of LOCC in general state preparation is largely unexplored. A systematic approach is still needed to find short-depth LOCC-assisted circuits for preparing general states, especially long-range entangled ones. In essence, this requires optimization over various LOCC-assisted circuits. Such an optimization task exists for scenarios without LOCC assistance and is often solved by variational quantum algorithms \cite{peruzzo_variational_2014, omalley_scalable_2016, kandala_hardware-efficient_2017, google_ai_quantum_and_collaborators_hartree-fock_2020, cerezo_variational_2021, ferguson_measurement-based_2021, guo_experimental_2024}. However, extending the variational toolkits to the LOCC case remains unexplored and faces potential challenges. It is unclear whether previous techniques, especially the quantum gradient protocol, are compatible with LOCC. Furthermore, the number of parameters that define complex LOCC protocols is large in general, which would induce computational inefficiency. Note that a naive approach may even introduce exponentially many parameters. More importantly, the trainability of variational LOCC-assisted circuits is an open problem. In variational algorithms, we frequently encounter barren plateaus. These refer to the phenomenon of gradients vanishing exponentially with the system size when the circuit depth increases \cite{mcclean_barren_2018}, worsened by physical noise accumulation \cite{wang_noise-induced_2021}. It is unclear how the introduction of LOCC would impact trainability: it may alleviate the depth problem but could also exacerbate barren plateaus without proper design.

\begin{figure}[tpb!]
    \centering
    \includegraphics[width=1.0\linewidth]{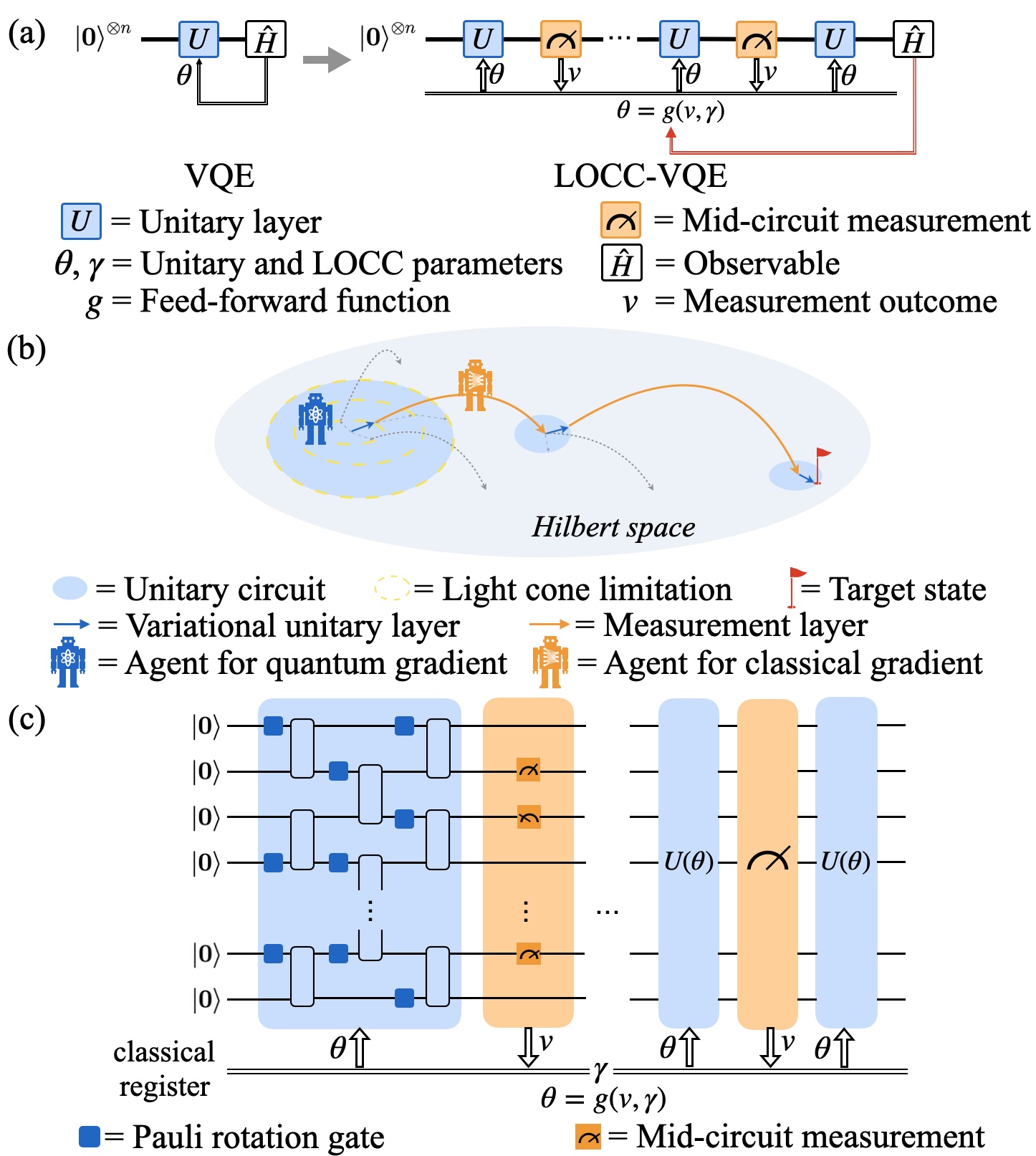}
    \caption{\justifying{LOCC-VQE scheme. Blue blocks represent unitary circuits, and orange blocks represent midcircuit measurements.  \textbf{(a)} Algorithm structure of LOCC-VQE. Gradient information is obtained for optimizing the LOCC parameters $\gamma$ in the feedback loop, represented by the red arrow. This feedback loop is the main difference compared to variational quantum algorithms. \textbf{(b)} Exploring the Hilbert space with LOCC-VQE. Among all possible paths, represented by dash arrows, agents obtained the gradient information to find an optimized state preparation path, represented by solid arrows, to reach the target state. As illustrated by the dotted yellow circles, LOCC enables jumps of states in the Hilbert space, breaking the light cone limitation on unitary circuits. \textbf{(c)} Variational LOCC-assisted quantum circuits. Parametrized unitary layers, represented in blue, and midcircuit measurement layers, represented in orange, are applied alternatively.}}
    \label{fig:loccvqe}
\end{figure}

In this work, we tackle these challenges by proposing the LOCC-assisted variational quantum eigensolver (LOCC-VQE) to solve the ground state of a given Hamiltonian, as depicted in Fig.~\ref{fig:loccvqe}.
To figure out the optimal LOCC protocol, we propose an efficient quantum-classical hybrid approach to estimate parameter gradients and present explicit and reasonable conditions for the absence of barren plateaus. Based on these gradients, we can perform gradient-based optimization to minimize the energy and solve the ground-state problem via LOCC-VQE, as depicted in Fig.~\ref{fig:loccvqe}(a). Notably, LOCC protocols can be selected with flexibility, allowing the incorporation of classical computations in various forms, such as look-up tables or neural networks. By choosing appropriate protocols, we provide LOCC-assisted advantages while ensuring trainability by avoiding barren plateaus, as depicted in Fig.~\ref{fig:loccvqe} (b).

We propose the following general parameterization of LOCC protocols: in a LOCC-assisted circuit, each unitary gate layer encompasses Pauli rotation gates with parameters $\theta$, as depicted in Fig.~\ref{fig:loccvqe} (c). These gate parameters are determined by classical protocol with measurement outcomes denoted by $\mbv$. The classical protocol is a function $g$ with LOCC parameters $\gamma$, by which gate parameters are computed as $\theta = g(\gamma, \mbv)$.
Note that our parameterization is not limited to any specific circuit architecture or LOCC protocol structures.

Since LOCC parameters $\gamma$ are independent variables that define the circuit, we will denote the output state as $\Psi_\gamma$. The state can be considered the mixture of postselected states with different midcircuit measurement outcomes,
\begin{equation}
    \Psi_\gamma = \sum_{\mbv} P_{\theta}(\mbv) \Phi_{\theta, \mbv}, \label{eq:decomp-normal}
\end{equation}
where $P_{\theta}(\mbv)$ is the probability of measurement outcome $\mbv$ and $\Phi_{\theta, \mbv}$ is the postmeasurement state.

Our goal is to find the optimal $\gamma$ that minimizes the energy of a given Hamiltonian $\hat{H}$, i.e., $\Tr \left[\hat{H} \Psi_{\gamma}\right]$. Note that ideally, when optimized to the ground state, different $\Phi_{\theta, \mbv}$ corresponding to different midcircuit measurement outcomes will be converted to the same pure ground state of the Hamiltonian $\hat{H}$.

To apply efficient gradient-based optimization, we will need the gradients $\nabla_{\gamma} \Tr \left[\hat{H} \Psi_{\gamma}\right]$. We first introduce the following proposition, which expresses these gradients as a quantum-classical composition.
\begin{proposition}
\label{prop:Quantum gradients for variational LOCC-assisted circuits}
    The gradients of a variational LOCC-assisted circuit can be obtained as the inner product of two matrices,
    \begin{equation}
        \frac{\partial \Tr \left[\hat{H} \Psi_{\gamma}\right]}{\partial \gamma_k} = \Tr[\left( \mathbf{G}^{C_k} \right)^T \mathbf{G}^{Q_k}] \label{eq:grad}.
    \end{equation}
    Here, the two matrices correspond to the quantum gradient,
    \begin{equation}
        \mathbf{G}^{Q_k} = \left\{\mathrm{g}^{Q_k}_{i,j} \right\} = \left\{\frac{\partial \Tr \left[\hat{H} P_{\theta}(\mbv_i) \Phi_{\theta, \mbv_i} \right]}{\partial \theta_j}\Bigg|_{\theta_j = g_j(\gamma, \mbv_i)} \right\}
    \end{equation}
    and classical gradients,
    \begin{equation}
    \mathbf{G}^{C_k} = \left\{\mathrm{g}^{C_k}_{i,j} \right\} = \left\{\frac{\partial g_j(\gamma, \mbv_i)}{\partial \gamma_k} \right\}
\end{equation}
 respectively, where $\mbv_i$ are measurement outcomes and $\theta_j$ are Pauli gate rotation angles in the circuit.
\end{proposition}

Based on this quantum-classical composition structure, we propose a hybrid protocol to estimate the gradient in Eq.~\eqref{eq:grad}. The approach consists of the following three steps: first, for estimation of the quantum part, $ \mathbf{G}^{Q_k}$, inspired by parameter shift rules \cite{peruzzo_variational_2014, omalley_scalable_2016, kandala_hardware-efficient_2017}, we shift each circuit parameter $\theta_j$ by $\pm \frac{\pi}{2}$, which gives the following shifted LOCC protocols
\begin{equation}
    g_{j\pm}(\gamma, \mbv) = g(\gamma, \mbv) \pm \frac{\pi}{2} \mbe_j,
\end{equation}
where $\mbe_j$ is the unit vector along the $j$th parameter direction, and estimate the expectation value of $\hat{H}$ by randomly picking a Pauli term from $\hat{H}$ and performing measurements on the corresponding basis. This measuring scheme provides an unbiased estimator of the energy expectation of the output state $\Phi_{\theta, \mbv}$,
i.e., $\langle \hat{H} \rangle_{j\pm,\mbv} = \Tr \left[\hat{H} \Phi_{\theta, \mbv}\right]\Big|_{\theta_j = g_{j\pm}(\gamma, \mbv)}$, where $\mbv$ is the measurement outcome. In each shot, we will store the midcircuit measurement results, $\mbv$, and the single-shot estimation of $\langle \hat{H} \rangle_{j\pm,\mbv}$.

Having the quantum measurement data from the first step, we perform the other two steps of classical processing. With each measured $\mbv$ as the input of the function $g_{j\pm}(\gamma, \mbv)$, we calculate the gradient $\frac{\partial g_j(\gamma, \mbv)}{\partial \gamma_k}$ for each circuit parameter $\theta_j$. Then, we reweight the previously obtained single-shot contributions by factors, $\pm \frac{\partial g_j(\gamma, \mbv)}{\partial \gamma_k}$, where a minus sign is assigned when $\theta_j$ is shifted by $-\frac{\pi}{2}$, and sum over the index $j$ and all samples $\mbv$. The details and proof of the algorithm are available in Appendix~\ref{app:grad}.

The advantage of this approach is its low sampling overhead. Naively, one may try postselection on the midcircuit measurement outcomes, $\mbv$, which leads to an exponentially large sample complexity. In our hybrid protocol design, we avoid the post-selection and ensure efficiency by introducing the aforementioned classical processing steps. Note that such an addition does not increase sampling overhead. Because our algorithm reuses sample data among different LOCC parameters $\gamma_k$, the sample complexity is only related to the number of tunable Pauli rotations in the circuit, which is the same for variational unitary circuits. Based on this gradient estimation algorithm, we can solve ground states via gradient-based optimization workflow illustrated in Fig.~\ref{fig:loccvqe} (a).

\noindent\textit{Conditions for nonvanishing gradients.}---
With LOCC-VQE, preparing long-range entangled states with a low circuit depth is possible. Such low depths appear promising for avoiding barren plateaus \cite{mcclean_barren_2018, wang_noise-induced_2021}. However, we still have to be cautious about how the additional LOCC components impact trainability. Here, we establish the following conditions under which gradients are nonvanishing, thereby ensuring trainability. The formal version and proof details are available in Appendix~\ref{appendix:absence of BP}.

\begin{theorem}[informal]
\label{thm:BP}
    The following conditions can ensure the gradients will not vanish as the number of qubits scales in variational LOCC-assisted circuits:
\begin{enumerate}
    \item[$\mathbf{\mathcal{A}1}.$] Hamiltonian is local.---The observable $\hat{H}$ is the sum of terms whose support has a constant size.
    \label{prop: local Hamiltonian}
    \item[$\mathbf{\mathcal{A}2}.$] The circuit depth is constant.
    \item[$\mathbf{\mathcal{A}3}.$] The gradient of the function $g$ will not exponentially decay as the size of its input increases.
    \item[$\mathbf{\mathcal{A}4}.$] Each LOCC protocol parameter $\gamma_j$ controls a constant number of quantum gates. The function $g$ has a constant support regarding each $\gamma_j$.
    \item[$\mathbf{\mathcal{A}5}.$] Each quantum gate parameter $\theta$ is controlled by a constant number of midcircuit measurement results.
\end{enumerate}
\end{theorem}

\begin{proof}[Proof sketch]
    The gradient of a LOCC parameter in Eq.~\eqref{eq:grad} is the inner product of two vectorized high-dimensional matrices, $\mathbf{G}^{Q_k}$ and $\mathbf{G}^{C_k}$. We first need to ensure these two matrices do not vanish individually. For $\mathbf{G}^{Q_k}$ not to vanish, we acquire conditions $\mathbf{\mathcal{A}1}$ and $\mathbf{\mathcal{A}2}$. These conditions can be understood from the perspective of an information propagation light cone, where information can only spread linearly in a geometrically local unitary circuit. Note that these conditions are also necessary to prevent the quantum gradient from vanishing exponentially in a unitary variational circuit, which is a special case of LOCC-assisted circuits. Similarly,  we introduce condition $\mathbf{\mathcal{A}3}$ to prevent the vanishing of classical gradients $\mathbf{G}^{C_k}$.

The above three conditions are not sufficient enough because the nonvanishing of individual components, $\mathbf{G}^{Q_k}$ and $\mathbf{G}^{C_k}$, does not guarantee the nonvanishing of the inner product. In fact, without further restraints, such an inner product will typically vanish due to the cancellations among degrees of freedom. Therefore, we introduce two additional conditions on the LOCC protocol, $\mathbf{\mathcal{A}4}$ and $\mathbf{\mathcal{A}5}$. These conditions pose a sparse structure in $\mathbf{G}^{C_k}$. By combining this with the sparsity in $\mathbf{G}^{Q_k}$, we can ensure that only a constant degree of freedom can contribute to the inner product, preventing the gradients from decaying as the number of qubits $n$ in the asymptotic limit.
\end{proof}

Theorem~\ref{thm:BP} ensures the large-scale effectiveness of our approach from a theoretical perspective. With nonvanishing gradients, the gradient-based optimization can find the optimal state that minimizes the energy. It is worth noting that Theorem~\ref{thm:BP} is stronger than the frequently referenced absence of barren plateaus \cite{mcclean_barren_2018, wang_noise-induced_2021, zhang_absence_2024}, which excludes exponential vanishing but still allows for vanishing to some polynomial scalings.

Our approach holds significant relevance for realistic experimental settings. While prior work has established the absence of barren plateaus in solving long-range entangled states in noise-free settings \cite{zhang_absence_2024}, noise-induced barren plateaus remain a challenge due to the deep circuit depth required \cite{aharonov_limitations_1996, wang_noise-induced_2021, yan_limitations_2023}. Our approach additionally prevents noise-induced barren plateaus through circuit depth reduction by LOCC. By minimizing circuit depth, our method enables the scalable exploration of long-range entangled states, even in the presence of noise.

\noindent\textit{One-dimensional chain models.}---
We employ LOCC-VQE to investigate systems exhibiting long-range entanglement through numerical simulations. Our simulations leverage the tensor-network-based circuit simulator introduced in \cite{zhang_tensorcircuit_2023}. For parameter optimization, we utilize the Adam optimizer \cite{kingma_adam_2015}. The quantum circuits consist of two-qubit gates parametrized using Cartan decomposition \cite{khaneja_cartan_2000, earp_constructive_2005}. The corresponding code is publicly accessible on GitHub \cite{loccvqegithub}.

Our first example is the parent Hamiltonian of the GHZ state with perturbations. The GHZ state is a long-range entangled state, whose parent Hamiltonian can be chosen as a one-dimensional (1D) Ising model, depicted in Fig.~\ref{fig:Qubit_lattices} (a), $\hat{H}_{\text{Ising}} = - \sum_{\langle i, j\rangle} Z_i Z_j$, where $\langle i, j\rangle$ represent the near-neighbor sites. The problem with the Ising model is its ground-state degeneracy. The ground-state subspace of the Ising model contains product states, which are trivially short-range entangled. To break this degeneracy and make the GHZ state the unique ground state, a term of $n$-qubit tensor product of $X$ operator is introduced, $- h\bigotimes_i X_i$, where $\hat{H}$ denotes the energy gap created above the GHZ state as the unique ground state. To exhibit the robustness of LOCC-VQE, we add perturbations in terms of Pauli operators on each site, resulting in
\begin{equation}
    \hat{H}_{\text{GHZ}} = - (1-\lambda)\sum_{\langle i, j\rangle}Z_iZ_j - (h-\lambda)\bigotimes_i X_i - \lambda\sum_i P_i,
    \label{eq:ghzpert}
\end{equation}
where $\lambda$ is the perturbation strength and $P_i\in{X_i,Y_i,Z_i}$ on the $i$th site. Here, we use the same Pauli operators to perturb all qubits, physically representing the direction of an external uniform field.

\begin{figure}[htbp!]
    \centering
    \includegraphics[width=1.0\linewidth]{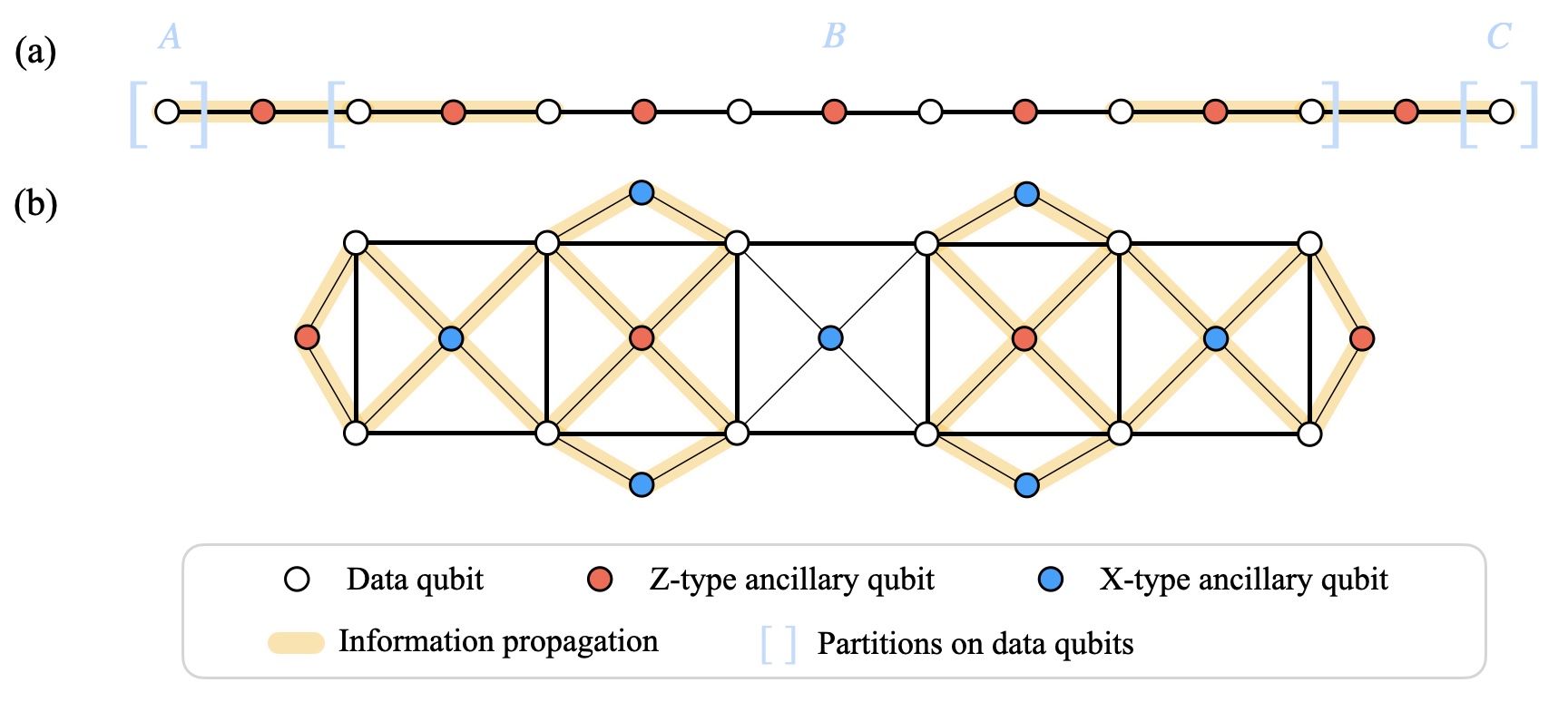}
    \caption{\justifying{Qubit layout for \textbf{(a)} the parent Hamiltonian of the GHZ state and \textbf{(b)} the surface code. To present the advantage of LOCC-VQE later, we define the following qubit partitions: in (a), light-blue regions $A$, $B$, and $C$ form a partition of the data qubits. In (b), the yellow region represents the light cone of information propagation through local two-qubit unitary gates.
    }}
    \label{fig:Qubit_lattices}
\end{figure}

In our numerical tests, we simulate an eight-qubit model, set $h=16$, and test various values of perturbation strength, $\lambda$. We first test the energy achieved by LOCC-VQE and unitary VQE of the same depth of $2$ with the perturbed Hamiltonian introduced above. For variational training, we set the same number of iterations. With variational LOCC-assisted circuit ansatz, information can propagate beyond the light cone limitation placed on unitary circuit ansatz, making it possible for long-range entanglement to emerge within shallow depth. The circuit design is inspired by the LOCC protocol for preparing a non-perturbed GHZ state as illustrated in \cite{piroli_quantum_2021}. For each pair of neighboring data qubits, we couple them with an ancilla qubit with parametrized gates and measure the ancilla qubits, the measurement outcomes are later fed into the LOCC protocol function $g$, detailed in Appendix~\ref{app:arch}.

\begin{figure}[htbp!]
    \centering
    \includegraphics[width=1.0\linewidth]{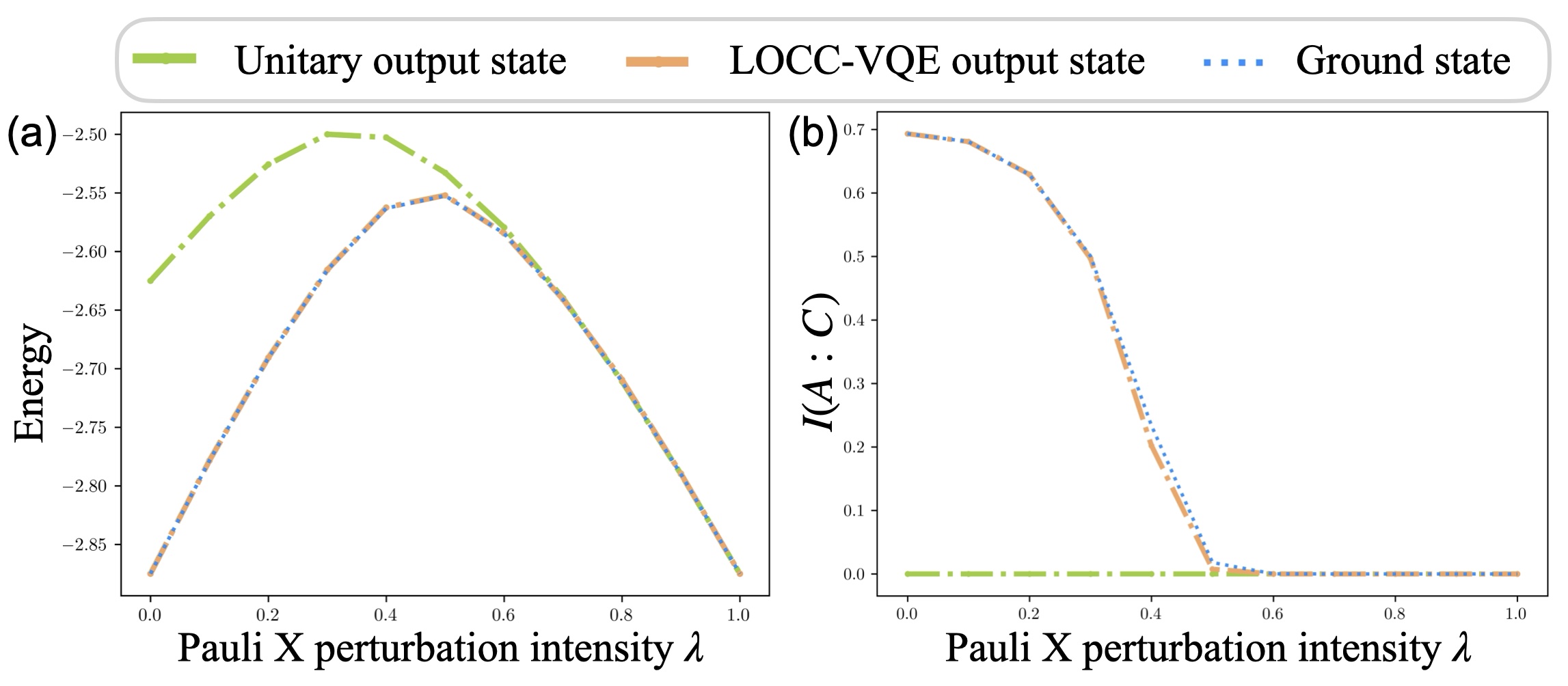}
    \caption{\justifying{Solving the parent Hamiltonian of the eight-qubit GHZ state with Pauli $X$ perturbations with depth-two circuits. \textbf{(a)} Comparison of the energy optimization results through LOCC-VQE and unitary VQE with depth-two circuits. \textbf{(b)} Comparison of the quantum mutual information (QMI) between subsystems, $I(A:C) = S(A) + S(C) - S(AC)$, where subsystems $A$ and $C$ are shown in Fig.~\ref{fig:Qubit_lattices}~(a) and $S(\cdot)$ denotes the von Neumann entropy. Here, we assume unlimited shots. Results with finite shots in midcircuit measurements are also studied and available in Appendix~\ref{app:suppnum}.}}
    \label{fig:GHZ_8_X_EandQMI}
\end{figure}

Our results suggest LOCC-VQE's advantages over its unitary counterpart with the same depth in predicting ground-state energy and QMI, as shown in Fig.~\ref{fig:GHZ_8_X_EandQMI}. For ground-state energy, LOCC-VQE can achieve a relative accuracy of $10^{-3}$ in any perturbation direction over the entire range of perturbation intensity, while unitary VQE can only achieve a relative accuracy of $10^{-1}$, as shown in Fig.~\ref{fig:GHZ_8_X_EandQMI}~(a). A precision gap of $2$ orders of magnitude between LOCC-VQE and unitary VQE is demonstrated when the perturbation intensity $\lambda$ in Eq.~\eqref{eq:ghzpert} is small, where long-range entanglement dominates the target ground state. To further demonstrate a provable advantage, we use QMI between subsystems $A$ and $C$ to characterize long-range entanglement and thereby separate LOCC-VQE and its unitary counterpart. Theoretically, given a unitary circuit depth, the QMI vanishes outside of the light cones of information propagation. In contrast, as shown in Fig.~\ref{fig:GHZ_8_X_EandQMI}(b), states prepared by LOCC-VQE have a nonzero QMI, the amount of which matches the ground state, while QMI is exactly zero for unitary VQE of the same circuit depth. This implies the advantage of LOCC-VQE originated from the ability to break the light cone of information propagation. The details of light-cone arguments are available in Appendix~\ref{appendix:QMI_def}.

Note that, strictly speaking, the $n$-qubit tensor product Puali-$X$ term in Eq.~\eqref{eq:ghzpert} does not satisfy the local Hamiltonian condition in Theorem~\ref{thm:BP}. Interestingly, our numerical results of LOCC-VQE can still prepare the ground state with high precision, as shown above. This implies LOCC-VQE's potential to work well even when conditions in Theorem~\ref{thm:BP} are relaxed. In a similar model without such a long-range term, the 1D transverse-field Ising model, we have also demonstrated the accurate results of the ground state preparation, whose Hamiltonian satisfies the local condition in Theorem~\ref{thm:BP}, as shown in Appendix~\ref{app:suppnum} along with other models of interest.

\noindent\textit{Perturbed rotated surface code.}---
The ground states of the surface code Hamiltonian also possess long-range entanglement, enabling logical information storage. In our numerical tests, we use the rotated surface code \cite{bombin_optimal_2007, kovalev_improved_2012, fowler_surface_2012, anderson_homological_2013, tomita_low-distance_2014}, which is a variant of Kiteav's toric code \cite{bravyi_quantum_1998, kitaev_fault-tolerant_2003} but with open boundary condition. We consider the perturbation of a magnetic field in the $Z$ direction, which results in the following Hamiltonian:
\begin{equation}
\label{eq: H_sur_perturbed}
    \hat{H}_{\text{sur}}(\lambda) = -(1-\lambda)\sum_v A_v - (1-\lambda)\sum_p B_p - \lambda \sum_{i=1}^{N_xN_y} Z_i.
\end{equation}
In this model, the qubits are arranged in a regular lattice, as shown in Fig.~\ref{fig:Qubit_lattices}(b). Here, $N_x$ and $N_y$ represent the width and height of the regular lattice, respectively. $A_v$ and $B_p$ are stabilizers for the unperturbed rotated surface code, while $\lambda$ represents the strength of the perturbation. The $Z$-type stabilizers $A_v$ and the $X$-type stabilizers $B_p$ are arranged in an alternating checkerboard pattern.
The purpose of highlighting $X$-type and $Z$-type ancillary qubits in Fig.~\ref{fig:Qubit_lattices}(b) is to illustrate the model better, and we do not distinguish them in our numerical experiments, treating them equally when initializing parameters. The robustness and flexibility of LOCC-VQE make achieving high precision in ground-state preparation possible without requiring prior knowledge of the type of ancillary qubits.

The results of preparing ground states of perturbed rotated surface code using LOCC-VQE are shown in Fig.~\ref{fig:Energy_surface_code}. LOCC-VQE can reach a relative error of $10^{-2}$ in energy precision for all perturbation intensity. Even with a limited lattice size, where the effect of long-range entanglement on the estimated ground state energy is not as strong as in larger lattice sizes, a precision gap of $3$ orders of magnitude between the ground states prepared by LOCC-VQE and unitary VQE is demonstrated, as shown in Fig.~\ref{fig:Energy_surface_code}, when the perturbation intensity $\lambda$ in Eq.~\eqref{eq: H_sur_perturbed} is small, demonstrating the advantage of LOCC-VQE.

In future work, various promising enhancements for LOCC-VQE will be explored. One could employ architecture search \cite{zhang_differentiable_2022} or further combine our approach with classical computations \cite{yuan_quantum_2021, huang_tensor-network-assisted_2023}. Theoretically, trainability is also worth further exploring. For example, one could relax the conditions in Theorem~\ref{thm:BP} to allow gradient vanishing to some extent while still preserving trainability.

\begin{figure}[htbp!]
    \centering
    \includegraphics[width=1.0\linewidth]{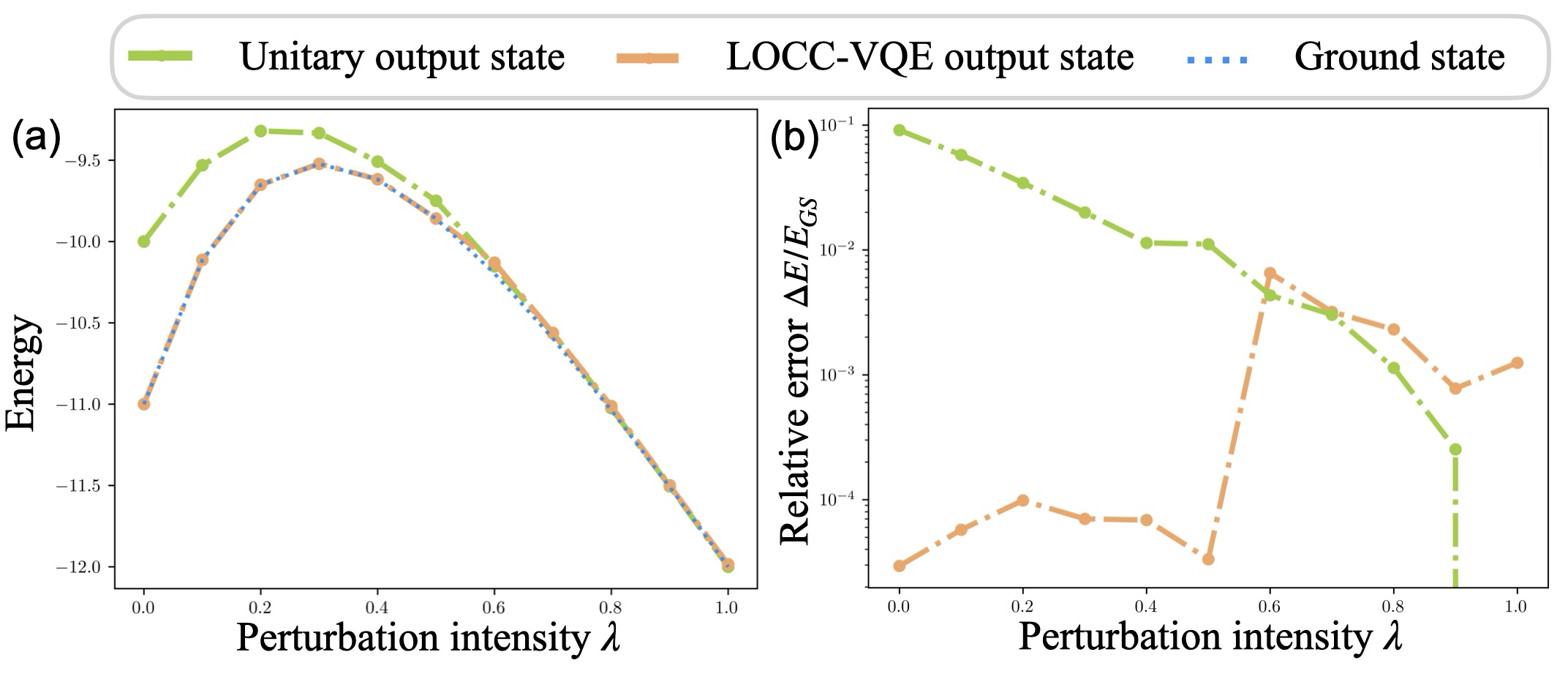}
    \caption{\justifying{Solving perturbed surface code. \textbf{(a)} Comparison between the energy optimization results through LOCC-VQE and unitary VQE with depth four circuits. \textbf{(b)} Comparison between the relative error of ground state energy optimization results, $\frac{\Delta E}{E_{GS}}=\frac{E-E_{GS}}{E_{GS}}$, through LOCC-VQE and unitary VQE with depth four circuits.}}
    \label{fig:Energy_surface_code}
\end{figure}

The experimental demonstration of LOCC-VQE has already been feasible on various quantum computing platforms with midcircuit measuring capabilities \cite{foss-feig_experimental_2023, iqbal_topological_2024, moses_race-track_2023, baumer_efficient_2023, corcoles_exploiting_2021}.  When running on a quantum computer, the computationally intensive part of LOCC-VQE will be greatly accelerated. We anticipate that experiments will unlock the full potential of LOCC-VQE in preparing long-range entangled states for quantum error correction, topological matters, and algorithms.

\begin{acknowledgments}
    We thank Luca Dellantonio, Zhenyu Du, Guoding Liu, Zhong-Xia Shang, Nathanan Tantivasadakarn, Xiao Yuan, and Shi-Xin Zhang for the helpful discussion. Y.~Y., M.~M., and X.~M. acknowledge the support of the National Natural Science Foundation of China (NSFC) Grants No.~12174216 and the Innovation Program for Quantum Science and Technology Grants No.~2021ZD0300804 and No.~2021ZD0300702. Y.~Z acknowledges the support of NSFC Grant No.~12205048, the Innovation Program for Quantum Science and Technology Grant No.~2021ZD0302000, Shanghai Science and Technology Innovation Action Plan Grant No.~24LZ1400200, and the start-up funding of Fudan University.
\end{acknowledgments}

\bibliography{./materials/loccvqe.bib}

\newpage

\onecolumngrid

\clearpage

\appendix

\section{Proofs and details of gradient estimation protocol}
\label{app:grad}

We start by revisiting the definition of LOCC-assisted circuits and formalizing it in the following way.

\begin{definition}[LOCC-assisted circuits] \label{def:adacirc}
Starting from the initial state $\ket{\Psi_0}$, we alternatively apply unitaries or measurements. Assumed that the outcomes are $\mbv = \{ v_j \}$, the unnormalized outcome state with respect to the outcome $\mbv = \{ v_j \}$ will be
\begin{equation}
    \ket{\tilde{\Phi}_{\mbv}} = U^{(d)}_\mbv \Pi^{(d-1)}_\mbv \cdots \Pi^{(1)}_\mbv U^{(1)} \ket{\Psi_0}. \label{eq:adapt}
\end{equation}
Here, $U^{(i)}_\mbv$ denotes unitaries, and $\Pi^{(i)}_\mbv$ denotes measurement projectors. Unitaries $U^{(i)}_\mbv$ may depend on earlier measurement outcomes corresponding to projectors $\Pi_j$ for $j<i$. The LOCC-assisted circuits, on average, will generate the following state:
\begin{equation}
    \Psi = \sum_{\mbv} \ketbra{\tilde{\Phi}_{\mbv}}{\tilde{\Phi}_{\mbv}}.
    \label{eq:mix}
\end{equation}
\end{definition}

The depth of LOCC-assisted circuits is the sum of all unitary layers, as defined below.

\begin{definition}[Depth of a LOCC-assisted circuit]
\label{def:depth}
    The depth of a LOCC-assisted circuit is the sum of the depths of $U^{(i)}_\mbv$ in Eq.~\eqref{eq:adapt}. The depth of each $U^{(i)}_\mbv$, denoted by $d_i$, is defined as the following: Decompose the unitary by $d_i$ layers of gates $\prod_{j=1}^{d_i} \bigotimes_{k} U^{(i)}_{j,k}$, where within a same layer, labelled by $j$, $U^{(i)}_{j,k}$ are two-qubit unitaries that do not overlap with each other.
\end{definition}

Our gradient estimation protocol is built on Proposition~1, for which we give proof below.

\begin{proof}[Proof of Proposition~1]
\begin{equation}
    \begin{split}
        \frac{\partial \Tr \left[\hat{H} \Psi_{\gamma}\right]}{\partial \gamma_j} & = \sum_{\mbv} \frac{\partial}{\partial \gamma_j} \Tr \left[\hat{H}\tilde{\Phi}_{\theta, \mbv}\right] \\
        & = \sum_{i,\mbv} \frac{\partial g_i(\gamma, \mbv)}{\partial \gamma_j} \frac{\partial \Tr \left[\hat{H}\tilde{\Phi}_{\theta, \mbv}\right]}{\partial \theta_i}\Bigg|_{\theta = g(\gamma, \mbv)} \\
        & = \sum_{i,\mbv} \frac{1}{2} \frac{\partial g_i(\gamma, \mbv)}{\partial \gamma_j} \left(\Tr \left[\hat{H} \tilde{\Phi}_{\theta, \mbv}\right]\Big|_{\theta = g_{i+}(\gamma, \mbv)} - \Tr \left[\hat{H} \tilde{\Phi}_{\theta, \mbv}\right]\Big|_{\theta = g_{i-}(\gamma, \mbv)}\right) \\
        & = \sum_{i,\mbv} \frac{1}{2} \frac{\partial g_i(\gamma, \mbv)}{\partial \gamma_j} \left(P_{\theta = g_{i+}(\gamma, \mbv)}(\mbv) \Tr \left[\hat{H} \Phi_{\theta, \mbv}\right]\Big|_{\theta = g_{i+}(\gamma, \mbv)} - P_{\theta = g_{i-}(\gamma, \mbv)}(\mbv) \Tr \left[\hat{H} \Phi_{\theta, \mbv}\right]\Big|_{\theta = g_{i-}(\gamma, \mbv)}\right).
    \end{split}
\end{equation}
The third equality uses parameter shifts for quantum gradients. Unlike unitary variational circuits, we adapt parameter shifts to circuits with projectors, as detailed in the lemma that follows.

It is more convenient to write the gradients into a matrix form. Denote the length of $\theta$ as $l$, and the length of $\mbv$ as $m$. Let the matrix representing the quantum gradients be $\mathbf{G}^{Q_k} \in \mathbb{R}^{l\times 2^m}$, calculated by parameter-shift, and the matrix representing the classical gradients be $\mathbf{G}^{C_k} \in \mathbb{R}^{l\times 2^m}$. Their expressions are
\begin{equation}
\label{eq:G_Q_and_G_C}
\begin{split}
    \mathbf{G}^{Q_k} &= \{\mathrm{g}^{Q_k}_{i,j}\} = \{\frac{\partial \Tr \left[\hat{H}\tilde{\Phi}_{\theta, \mbv_i}\right]}{\partial \theta_j}\Bigg|_{\theta_j = g_j(\gamma, \mbv_i)}\}, \\
    \mathbf{G}^{C_k} &= \{\mathrm{g}^{C_k}_{i,j}\} = \{\frac{\partial g_j(\gamma, \mbv_i)}{\partial \gamma_k}\} ,
\end{split}
\end{equation}
where $\mbv_i$ is a bit string of length $m$ representing the $i^{th}$ possible way of projectors.

We can rewrite the quantum gradient for variational LOCC-assisted circuits as
\begin{equation}
\label{eq:Hadamard product form of the gradients}
    \frac{\partial \Tr \left[\hat{H} \Psi_{\gamma}\right]}{\partial \gamma_k} =  \langle\mathbf{G}^{C_k}, \mathbf{G}^{Q_k} \rangle_F =  \Tr[\left( \mathbf{G}^{C_k} \right)^T \mathbf{G}^{Q_k}]
\end{equation}
where $\langle\cdot, \cdot \rangle_F$ is the Frobenius inner product, which can be seen as the inner product of the vectorized representation of $\left(\mathbf{G}^{C_k}\right)^T$ and $\left(\mathbf{G}^{Q_k}\right)^T$.
\end{proof}

\begin{lemma}[Parameter shifts with midcircuit measurements] \label{lem:grad}
\begin{equation}
    \frac{\partial \Tr \left[\hat{H}\tilde{\Phi}_{\theta, \mbv}\right]}{\partial \theta_i} = \frac{1}{2} \left(\Tr \left[\hat{H} \tilde{\Phi}_{\theta + \frac{\pi}{2} \mbe_i, \mbv}\right] - \Tr \left[\hat{H} \tilde{\Phi}_{\theta - \frac{\pi}{2} \mbe_i, \mbv}\right]\right).
\end{equation}
\end{lemma}
\begin{proof}
    By Eq.~\eqref{eq:adapt},
    \begin{equation}
        \tilde{\Phi}_{\theta, \mbv} = U^{(d)} \Pi^{(d-1)} \cdots \Pi^{(1)} U^{(1)} \ket{\Psi_0} \bra{\Psi_0} {(U^{(1)})}^\dagger \Pi^{(1)} \cdots \Pi^{(d-1)} {(U^{(d)})}^\dagger,
    \end{equation}
    where we omit the subscript $\mbv$ for simplicity.

    Express $U^{(k)}$ as $U^{(k)} = V W(\theta_i) V^\prime$, where $W(\theta_i) = e^{-i\frac{\theta_i}{2}\Sigma_i}$, $\Sigma_i$ is a Pauli operator and $\theta_i$ is assumed to act non-trivially on $U^{(k)}$. Note that
    \begin{equation}
        \begin{split}
            \frac{\partial U^{(k)}}{\partial \theta_i} &= -\frac{i}{2} V \Sigma_iW(\theta_i) V^\prime, \\
            \frac{\partial (U^{(k)})^{\dagger}}{\partial \theta_i} &= \frac{i}{2} {(V^\prime)}^\dagger \Sigma_i W(-\theta_i) V^\dagger.
        \end{split}
    \end{equation}
    Then, we calculate the gradient,
    \begin{equation}
        \begin{split}
            \frac{\partial \Tr \left[\hat{H}\tilde{\Phi}_{\theta, \mbv}\right]}{\partial \theta_i} &=  \Tr \left[\hat{H} U^{(d)} \Pi^{(d-1)} \cdots \frac{\partial U^{(k)}}{\partial \theta_i} \cdots \Pi^{(1)} U^{(1)} \ket{\Psi_0} \bra{\Psi_0} {(U^{(1)})}^\dagger \Pi^{(1)} \cdots \Pi^{(d-1)} {(U^{(d)})}^\dagger + h.c.\right] \\
            &= -\frac{i}{2} \Tr \left[\hat{H} U^{(d)} \Pi^{(d-1)} \cdots V W(\theta_i) \left[ \Sigma_i, \rho \right] W(-\theta_i) V^\dagger \cdots \Pi^{(d-1)} {(U^{(d)})}^\dagger \right].
        \end{split}
    \end{equation}
    where $\rho = V^\prime U^{(i-1)} \Pi^{(i-2)} \cdots \Pi^{(1)} U^{(1)} \ket{\Psi_0} \bra{\Psi_0} {(U^{(1)})}^\dagger \Pi^{(1)} \cdots \Pi^{(i-2)} {(U^{(i-1)})}^\dagger {(V^\prime)}^\dagger$.

    Further, using the fact that $\left[ \Sigma_j, \rho \right] = i \left[ W(\frac{\pi}{2}) \rho W(\frac{\pi}{2})^\dagger - W(-\frac{\pi}{2}) \rho W(-\frac{\pi}{2})^\dagger \right]$, we have
    \begin{equation}
        \begin{split}
             \frac{\partial \Tr \left[\hat{H}\tilde{\Phi}_{\theta, \mbv}\right]}{\partial \theta_i} &= -\frac{i}{2} \Tr \left[\hat{H} U^{(d)} \Pi^{(d-1)} \cdots V W(\theta_i) \left[ W(\frac{\pi}{2}) \rho W(\frac{\pi}{2})^\dagger - W(-\frac{\pi}{2}) \rho W(-\frac{\pi}{2})^\dagger \right] W(-\theta_i) V^\dagger \cdots \Pi^{(d-1)} {(U^{(d)})}^\dagger \right] \\
             &= U^{(d)} \Pi^{(d-1)} \cdots \Pi^{(1)} U^{(1)} \ket{\Psi_0} \bra{\Psi_0} {(U^{(1)})}^\dagger \Pi^{(1)} \cdots \Pi^{(d-1)} {(U^{(d)})}^\dagger.
        \end{split}
    \end{equation}
\end{proof}

Compared with the usual quantum gradient, the lemma considers circuits with projectors as the consequence of measurements. The proof above shows that projectors do not affect the quantum gradients for unitary circuits.

As explained in the main text, Proposition~1 implies the gradient can be obtained as a combination of quantum and classical parts, which can be estimated via a quantum-classical hybrid algorithm. We provide details in Algorithm~\ref{alg:grad}.

\SetKwComment{Comment}{/* }{ */}

\begin{algorithm}
\caption{Gradient estimation protocol for LOCC-VQE} \label{alg:grad}
\KwData{Observable $\hat{H}$; ansatz $\Psi_\gamma$ defined by $g(\gamma, \mbv)$; estimation sample rounds $M$.}
\KwResult{Estimated gradient $\{G_j\}_{j=1,\cdots,|\gamma|}$.}

\For{$i\leftarrow 1$ \KwTo $|\theta|$}{
    $g_{i\pm}(\gamma, \mbv) \gets g(\gamma, \mbv) \pm \frac{\pi}{2} \mbe_i$\;
    $\mathcal{C}_{i+} \gets \emptyset$\;
    \For{$k\leftarrow 1$ \KwTo $M$ \Comment*[r]{$|\cdot|$ denotes the parameter vector length}}
    {
        Run the LOCC-assisted circuit using $g_+$\Comment*[r]{quantum computer}
        $\mbv \gets$ midcircuit measurement results\;
        $c \gets$ one-shot estimation of $\hat{H}$ using $g_+$\;
        Add the pair, $(\mbv, c)$, to $\mathcal{C}_{i+}$\;
    }
    Do the same procedure to get $\mathcal{C}_{i-}$ from $g_-$\;
}

\For{$j\leftarrow 1$ \KwTo $|\gamma|$}{
    $G_+ \gets 0$\;
    $G_- \gets 0$\;
    \For{$i\leftarrow 1$ \KwTo $|\theta|$}{
        \For{$(\mbv, c) \in \mathcal{C}_{i+}$}{
            $G_+ \gets G_+ + \frac{1}{2} \frac{\partial g_i(\gamma, \mbv)}{\partial \gamma_j} c$\;
        }
        \For{$(\mbv, c) \in \mathcal{C}_{-}$}{
            $G_- \gets G_+ + \frac{1}{2} \frac{\partial g_i(\gamma, \mbv)}{\partial \gamma_j} c$\;
        }
    }
     $G_j \gets \frac{1}{M |\theta|} (G_+ - G_-)$\;
}
\end{algorithm}

\section{Proof of absence of baren plateaus}
\label{appendix:absence of BP}

In this section, we show the proof of the absence of baren plateaus, as suggested by Theorem~1 in the main text. First, we give its formal version below and summarize the proof in Fig.~\ref{fig:proof}.

\begin{theorem*}[Formal version of Theorem~1]
\label{thm:formal_thm_1}
    We assume that the vectorized representation of $\left(\mathbf{G}^{C_k}\right)^T$ and $\left(\mathbf{G}^{Q_k}\right)^T$ are uniformly distributed on the unit sphere in $\mathbb{R}^\mathcal{D}$ where $\mathcal{D}$ is the dimension of the parameters space.
    Then, the following conditions can ensure the gradients $\frac{\partial \Tr \left[\hat{H} \Psi_{\gamma}\right]}{\partial \gamma_k} = \Tr[\left( \mathbf{G}^{C_k} \right)^T \mathbf{G}^{Q_k}]$ will not vanish as the number of qubits scales in variational LOCC-assisted circuits:
\begin{enumerate}
    \item[$\mathbf{\mathcal{A}1}.$] Hamiltonian is local.---The observable $\hat{H}$ is the sum of terms whose support has a constant size.
    \item[$\mathbf{\mathcal{A}2}.$] The circuit depth is constant.
    \item[$\mathbf{\mathcal{A}3}.$] The gradient of the function $g$ will not exponentially decay as the size of its input increases.
    \item[$\mathbf{\mathcal{A}4}.$] Each LOCC protocol parameter $\gamma_j$ controls a constant number of quantum gates.---The function $g$ has a constant support regarding each $\gamma_j$.
    \item[$\mathbf{\mathcal{A}5}.$] Each quantum gate parameter $\theta$ is controlled by a constant number of midcircuit measurement results.
\end{enumerate}

\end{theorem*}

The assumption of uniform distributions is made for technical simplicity. The choice of the initial distribution of parameters varies and case-by-case analysis is impractical.
Notice that, regarding the inner product structure as shown in Proposition~1, this technical assumption is reasonable for proving the non-decaying gradients in LOCC-VQE since the inner product of two vectors uniformly sampled from the high-dimensional sphere decays exponentially with the dimension of the space in expectation.
For simplicity, we also focus on the case that $\hat{H}$ has only one term in the following proof. We can add them if multiple terms are in $\hat{H}$.

\begin{figure}[hbtp!]
	\begin{tikzpicture}[
		scale=1,
		every text node part/.style={align=center},
		every rectangle node/.style={rounded corners},
		>=latex
		]
		\node[rectangle,draw,minimum width=3.5cm, minimum height=1.6cm, text centered, anchor=north west, dashed]
        (condition_1) at (0,0)
        {Local Hamiltonian\\(Condition~$\mathbf{\mathcal{A}1}$)};

		\node[rectangle,draw,minimum width=3.5cm, minimum height=1.6cm, text centered, anchor=north west, dashed]
        (condition_2) at ($ (condition_1.north east) + (0.3cm,0) $)
        {Constant-depth\\(Condition~$\mathbf{\mathcal{A}2}$)};

		\node[rectangle,draw,minimum width=3.5cm, minimum height=1.6cm, text centered, anchor=north west, dashed]
        (condition_3) at ($ (condition_2.north east) + (0.3cm,0) $)
        {Non-trivial\\ classical gradients\\(Condition~$\mathbf{\mathcal{A}3}$)};

		\node[rectangle,draw,minimum width=3.5cm, minimum height=1.6cm, text centered, anchor=north west, dashed]
        (condition_4) at ($ (condition_3.north east) + (0.3cm,0) $)
        {Bounded support $\gamma_j$\\(Condition~$\mathbf{\mathcal{A}4}$)};

		\node[rectangle,draw,minimum width=3.5cm, minimum height=1.6cm, text centered, anchor=north west, dashed]
        (condition_5) at ($ (condition_4.north east) + (0.3cm,0) $)
        {Bounded support $\theta$\\(Condition~$\mathbf{\mathcal{A}5}$)};

		\node[rectangle,draw,minimum width=3.5cm, minimum height=1.6cm, text centered, anchor=north west]
        (lemma_2) at ($ (condition_1.south west) + (0,-2cm) $)
        {Non-trivial\\ quantum gradients\\(Lemma~\ref{lemma: entries of G^Q is not exponentially small})};

		\node[rectangle,draw,minimum width=3.5cm, minimum height=1.6cm, text centered, anchor=north west]
        (lemma_3) at ($ (condition_2.south west) + (0,-2cm) $)
        {Constant projectors\\in light cone\\(Lemma~\ref{lem:constant projector})};

		\node[rectangle,draw,minimum width=3.5cm, minimum height=1.6cm, text centered, anchor=north west]
        (lemma_4) at ($ (lemma_3.south west) + (0,-2cm) $)
        {Bounded \# of different\\ quantum gradients\\(Lemma~\ref{lem:sparsity of G^Q})};

		\node[rectangle,draw,minimum width=3.5cm, minimum height=1.6cm, text centered, anchor=north west]
        (lemma_5) at ($ (condition_4.south west) + (0,-2cm) $)
        {Constant non-zero \\rows in $\mathbf{G}^{C_k}$\\(Lemma~\ref{lem:row sparsity of G^C})};

		\node[rectangle,draw,minimum width=3.5cm, minimum height=1.6cm, text centered, anchor=north west]
        (lemma_6) at ($ (condition_5.south west) + (0,-2cm) $)
        {Bounded \# of different\\classical gradients\\(Lemma~\ref{lem: column sparsity of G^C})};

		\node[rectangle,draw,minimum width=7.3cm, minimum height=1cm, text centered, anchor=north west, line width=0.7mm]
        (G_Q) at ($ (lemma_2.south west) + (0cm,-5.6cm) $)
        {Sparse \& non-trivial $\mathbf{G}^{Q_k}$ \eqref{eq:G_Q_and_G_C}};

		\node[rectangle,draw,minimum width=11.1cm, minimum height=1cm, text centered, anchor=north west, line width=0.7mm]
        (G_C) at ($ (G_Q.north east) + (0.3cm,0) $)
        {Sparse \& non-trivial $\mathbf{G}^{C_k}$ \eqref{eq:G_Q_and_G_C}};

        \node[rectangle,draw,minimum width=3.5cm, minimum height=1.6cm, text centered, anchor=north west]
        (lemma_1) at ($ (G_Q.south west) + (0, -2cm) $)
        {$\langle\mathbf{G}^{C_k}, \mathbf{G}^{Q_k} \rangle_F$\\(Lemma~\ref{lem:grad})};

		\node[rectangle,draw,minimum width=3.5cm, minimum height=1.6cm, text centered, anchor=north east]
        (lemma_7) at ($ (G_C.south east) + (0, -2cm) $)
        {High-dimensional\\inner product\\(Lemma~\ref{lem: high-dimensional inner product})};

		\node[rectangle,draw,minimum width=3.5cm, minimum height=1.6cm, text centered, anchor=north, line width=1mm]
        (result) at ($ (condition_3.south) + (0, -12.2cm) $)
        {Nonvanishing gradient\\$\frac{\partial \Tr \left[\hat{H} \Psi_{\gamma}\right]}{\partial \gamma_k}$};

        \draw[->] (condition_1.south) -- (lemma_3.north);
        \draw[->] (condition_1.south) -- (lemma_2.north)
        node[pos=0.3, fill=white, text=black]{$S\coloneq\operatorname{supp}(\hat{H})$}
        node[pos=0.7, fill=white, text=black]{$\kappa = |S| = \mathcal{O}(1)$};

        \draw[->] (condition_2.south) -- (lemma_2.north)
        node[midway, fill=white, text=black] {Backward light cone $\mathcal{L}_{S}$};

        \draw[->] (condition_2.south) -- (lemma_3.north);

        \draw[->] (lemma_2.south) -- ($ (G_Q.north west) + (1.725cm, 0) $);

        \draw[->] (lemma_3.south) -- (lemma_4.north)
        node[midway,right] {Definition~\ref{def:light_cone_area} \\ Remark~\ref{rem:projector_and_light_cone} \\ Remark~\ref{rem:feedforward_and_light_cone}}
        node[midway,left]{Projectors\\ \& \\information\\propagation\\ light cone};

        \draw[->] (lemma_4.south) -- ($ (G_Q.north east) + (-1.725cm, 0) $)
        node[midway, fill=white, text=black]{Row-wise sparsity \\$\zeta_i = \mathcal{O}(1)$};

        \draw[->] (condition_3.south) -- ($ (G_C.north west) + (1.725cm, 0) $);

        \draw[->] (condition_4.south) -- (lemma_5.north)
        node[midway,fill=white, text=black]{$\iota_k = |\operatorname{supp}(\gamma_k)|$\\if $\theta_j \notin \operatorname{supp}(\gamma_k)$, $\mathrm{g}^{C_k}_{i,j}=0,\forall i$};

        \draw[->] (condition_5.south) -- (lemma_6.north)
        node[midway,fill=white, text=black]{$\nu_j= \mathcal{O}(1)$};

        \draw[->] (lemma_5.south) -- (G_C.north)
        node[midway, fill=white, text=black]{Row-wise sparsity \\$\iota_k = \mathcal{O}(1)$};

        \draw[->] (lemma_6.south) -- ($ (G_C.north east) + (-1.725cm, 0) $)
        node[midway, fill=white, text=black]{Column-wise sparsity \\$2^{\nu_j} = \mathcal{O}(1)$};

        \draw[->] (lemma_7.west) -- (result.east)
        node[midway,above]{Anti-concentration}
        node[midway, below]{inequality};

        \draw[->] (lemma_1.east) -- (result.west)
        node[midway,above]{Inner product}
        node[midway, below]{structure};

        \draw[->] (G_Q.south) -- (result.north);
        \draw[->] (G_C.south) -- (result.north);

        \node[text centered, anchor = south, fill=white, text=black]
        (dim) at ($ (result.north) + (0, 0.4cm) $)
        {$\mathcal{D} \leq \mathcal{O}(\iota_k \cdot (\max_i(\zeta_i)+\max_j(2^{\nu_j})) = \mathcal{O}(1)$};
	\end{tikzpicture}
	\caption{The flowchart of the nonvanishing gradients analysis (proof of Theorem~1).}
    \label{fig:proof}
\end{figure}

\subsection{Locality conditions and backward light cone}

Conditions that we introduce in Theorem~1 can be understood as various forms of locality in the system. Here, we first revisit the conditions in the main text one by one and intuitively explain their relation to trainability.

The condition $\mathbf{\mathcal{A}1}$ in Theorem~1 set each term in $\hat{H}$ only affect at most $\kappa=\mathcal{O}(1)$ local qubits. If $\hat{H} = \hat{H}_1+\hat{H}_2+\dots +\hat{H}_M$, $\kappa \coloneq \max(|\operatorname{supp}(\hat{H}_i)|) = \mathcal{O}(1)$.

The condition $\mathbf{\mathcal{A}2}$ in states that the LOCC circuit has constant depth $d=\mathcal{O}(1)$. Here, the depth of the circuit follows Definition~\ref{def:depth}. The condition $\mathbf{\mathcal{A}3}$ ensures the classical gradients do not vanish. The last two conditions $\mathbf{\mathcal{A}4}$ and $\mathbf{\mathcal{A}5}$ are less intuitive at first encounter. It means that for all $l$ outputs of the LOCC protocol function $g$, each parameter $\gamma$ at most affects $\iota = \mathcal{O}(1)$ of them. Meanwhile, each of the $l$ outputs is only controlled by $\nu=\mathcal{O}(1)$ input measurement results.

These conditions are related to the light cone perspective of how information propagates through quantum circuits, which is defined as follows:

\begin{definition}[Information propagation light cone in quantum circuit]
\label{def:light_cone}
    In a quantum circuit, the information propagation light cone (or the backward light cone) $\mathcal{L}$ with respect to an observable $O$ is defined as the set of gates whose change is possible to make a difference on the expectation value of $O$ based on the geometric structure of the circuit. From the dynamical perspective, one can view the information propagation light cone as the set of gates included in the backward evolution of $O$ in the Heisenberg picture, as depicted in Fig.~\ref{fig: Backward light cone}.

    To be more specific, consider a quantum circuit $U$ with input $\ket{\psi_i}$ and output $\ket{\psi_o}$, i.e. $\ket{\psi_o} = U\ket{\psi_i}$. The expectation of the $O$ is $\expval{O} = \bra{\psi_o}O\ket{\psi_o} = \bra{\psi_i}U^{\dagger}OU\ket{\psi_i}$, which can be viewed as taking the expectation value of the observable $\Tilde{O} = U^{\dagger} O U$ on the input state $\ket{\psi_i}$. Then the information propagation light cone $\mathcal{L}_{U,O}$ is the irreducible set of gates contributing to $\Tilde{O}$, with respect to circuit $U$ and observable $O$.
\end{definition}

\begin{figure}[htpb!]
    \centering
    \includegraphics[width=0.6\linewidth]{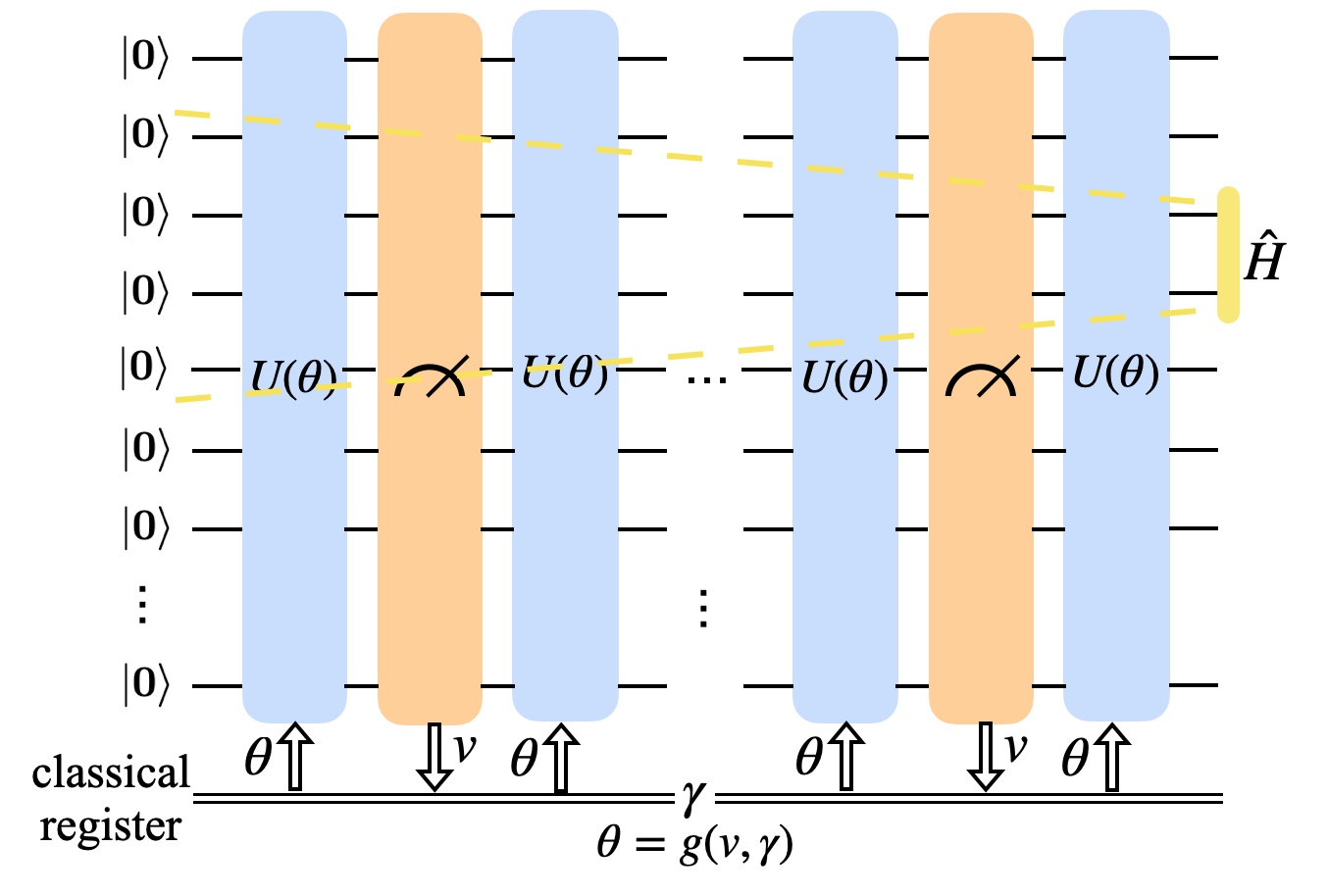}
    \caption{Backward light cone of a local Hamiltonian. The dashed yellow line represents the backward light cone of a local Hamiltonian $\hat{H}$, with the support of $\hat{H}$ represented by the vertical yellow bar, containing a constant number of qubits.}
    \label{fig: Backward light cone}
\end{figure}

\begin{definition}[Set of qubits in the light cone]
\label{def:qubits_in_light_cone}
    The set of qubits $\mathcal{Q}$ is in the information propagation light cone $\mathcal{L}_{U,O}$ with respect to a quantum circuit $U$ and the observable $O$ if for any qubit $q\in\mathcal{Q}$, there exist a gate $g\in\mathcal{L}_{U,O}$ such that $q$ is in the support of $g$.
\end{definition}

\begin{definition}[Area covered by the light cone]
\label{def:light_cone_area}
    For a geometrically local observable, The area covered by the information propagation light cone in a quantum circuit is the set of all elements (gates, projectors, etc.) in the circuit included in the geometric convex hull of the information propagation light cone as defined in Definition~\ref{def:light_cone}. For a general observable $O$, it is always possible to write it as the product of $m\leq n$ geometrically local observables $\{O_j\}_{j=1}^m$, i.e. $O = O_1O_2\ldots O_m$, then the area covered by the information propagation light cone of observable $O$ is the union of all areas covered by the $m$ information propagation light cones of observables $O_j$.
\end{definition}

We leave the following remarks concerning these definitions:

\begin{remark}
\label{rem:geo_local_qubits_in_light_cone}
    For geometrically local quantum circuits $U$ of depth $d$, the size of any set of qubits in the information propagation light cone $\mathcal{L}_{U,O}$ of observable $O$ has $\mathcal{O}(d)$ scaling with the circuit depth.
\end{remark}

\begin{remark}
\label{rem:projector_and_light_cone}
    Projectors in a quantum circuit will not enlarge the information propagation light cone since it will not spread information.
\end{remark}

\begin{remark}
\label{rem:feedforward_and_light_cone}
    When we are considering the parameter of a gate and taking \emph{partial} derivatives over them to calculate the gradient in Proposition~1, the projectors and feed-forward controls will not enlarge the information propagation light cone since they are all fixed by the definition of taking partial derivatives.
\end{remark}

With the conditions in Theorem~1, we can prove sparse structures of quantum and classical gradient matrices with light cone arguments. These sparsity structures will lead to proving the number of free parameters $\mathcal{D}=\mathcal{O}(1)$ and further overcoming the tendency of having vanishing values induced by the inner product structure when computing the gradients in Eq.\eqref{eq:Hadamard product form of the gradients}, eventually proving a lower bound for the gradient of LOCC-VQE to be a constant independent of $n$.

\subsection{Proof of Theorem~1}

As illustrated in Fig.~\ref{fig:proof}, our proof consists of the following lemmas.

\begin{lemma}
\label{lemma: entries of G^Q is not exponentially small}
Under conditions in Theorem~1, each entry of $\mathbf{G}^{Q_k}$ does not decay with $n$ in the asymptotic limit.
\end{lemma}

\begin{proof}
    In Lemma~\ref{lem:grad}, it has been proved that the quantum gradient can still be calculated with parameter shifts. With condition $\mathbf{\mathcal{A}1}$ and $\mathbf{\mathcal{A}2}$, the circuit depth is constant, and the observable is local.
    Consider each entry $\mathbf{G}^{Q_k}_{ij}$ as a function of $n$, and the lemma can be interpreted as all these functions do not decay with $n$ in the asymptotic limit. To see this, notice that the gradient is calculated based on the local observable with constant support, thus the information propagation light cone of this observable, as defined in Definition~\ref{def:light_cone}, and the size of the qubits in the information propagation light cone, as defined in Definition~\ref{def:qubits_in_light_cone}, are both constant since the circuit depth is constant. Thus, tracing out the subsystem not in the information propagation light cone, resulting in a system of \emph{constant} size, will not change the output of the aforementioned functions. This indicates that the results of these aforementioned functions do not depend on the system size. As a remark, for the entries related to gates out of the information propagation light cone, they will always be zero, which naturally will not decay with the system size in the asymptotic limit.
\end{proof}

\begin{lemma}
\label{lem:constant projector}
    Under conditions in Theorem~1, the number of projectors in the area covered by the information propagation light cone of $\hat{H}$ is constant.
\end{lemma}
\begin{proof}
    The proof is based on condition $\mathbf{\mathcal{A}1}$ and $\mathbf{\mathcal{A}2}$ with the light cone argument as illustrated in Fig.~\ref{fig: Backward light cone}. Denote the support of the observable $\hat{H}$ as $S\coloneq\operatorname{supp}(\hat{H})$. Condition $\mathbf{\mathcal{A}1}$ states that $\kappa = |S| = \mathcal{O}(1)$ and $S$ follows the locality constraint. Denote the backward light cone of $S$ as $\mathcal{L}_{S}$.

    First, we consider the case where $S$ is geometrically local to simplify notations. In that case, the lemma is a straightforward consequence for a constant depth circuit with a local observable $\hat{H}$.

    In general, if $S$ is not geometrically local, $S$ will still be the union of a constant number of geometrically local sets, with each set containing only a constant number of elements due to the constant size property of $S$. For each of these local sets, we can use the above arguments to prove a constant number of projectors. When added up, the number of projectors in the information propagation light cone $\mathcal{L}_{S}$ is constant.

\end{proof}

\begin{lemma}
\label{lem:sparsity of G^Q}
    Under conditions in Theorem~1, the number of different values in each row of $\mathbf{G}^{Q_k}$ is constant.
\end{lemma}
\begin{proof}
    For each quantum parameter $\theta_i$, corresponding to the $i$th row $\mathbf{G}^{Q_k}$, we are doing partial derivative of $\partial \Tr \left[\hat{H} \Psi_{\gamma}\right]$ over $\theta_i$. This means that we need to fix all other quantum parameters during the calculation. Consider the entries in the $i$th row of $\mathbf{G}^{Q_k}$, each entry corresponds to the gradient value when obtaining a possible measurement result. However, with all other parameters fixed, only the measurement outcomes in the backward light cone can affect the expectation value of $\hat{H}$. In other words, only the different measurement outcomes within the backward light cone of $\hat{H}$ can contribute to having different values in the $i$th row of $\mathbf{G}^{Q_k}$.

    Denote the $i^{th}$ row of $\mathbf{G}^{Q_k}$ as $\mathbf{g}^{Q_k}_i$, and let $\tilde{g}^{Q_k}_i$ be the set of all different values in $\mathbf{g}^{Q_k}_i$. Denote the size of $\tilde{g}^{Q_k}_i$ as $\zeta_i \coloneqq |\tilde{g}^{Q_k}_i|$.

    With lemma~\ref{lem:constant projector}, only a constant number of projectors affect the expectation value of the observable $\hat{H}$. Denote the number of projectors that affect the expectation value of the observable $\hat{H}$ as $\chi = \mathcal{O}(1)$. Then, for any row $i$ of $\mathbf{G}^{Q_k}$, $\zeta_i$ is upper bounded by $2^{\chi}$. Thus, we have
    \begin{equation}
        \zeta_i \leq 2^{\chi} = \mathcal{O}(1), \forall i\in[l],
    \end{equation}
    which completes the proof.
\end{proof}

\begin{lemma}
\label{lem:row sparsity of G^C}
    Under conditions in Theorem~1, the number of non-zero rows of $\mathbf{G}^{C_k}$ is constant.
\end{lemma}
\begin{proof}
    Denote the support of the parameter $\gamma_k$ as $\operatorname{supp}(\gamma_k)$:
    \begin{equation}
        \operatorname{supp}(\gamma_k) = \{\theta_{\gamma_k}^1, \{\theta_{\gamma_k}^2, \dots,\{\theta_{\gamma_k}^{\iota}\}
    \end{equation}
    which is a subset of the outputs $\theta$ of the LOCC protocol function $g$. The size of $\operatorname{supp}(\gamma_k)$ is $\iota_k = |\operatorname{supp}(\gamma_k)| = \mathcal{O}(1)$. For each output $\theta_j = g_j(\gamma, \mbv_i)$, if $\theta_j \notin \operatorname{supp}(\gamma_k)$, we have:
    \begin{equation*}
        \mathrm{g}^{C_k}_{i,j} = \frac{\partial \theta_j}{\partial \gamma_k} = \frac{\partial g_j(\gamma, \mbv_i)}{\partial \gamma_k} = 0, \forall i\in[2^m].
    \end{equation*}
    Thus, all rows of $\mathbf{G}^{C_k}$ that do not correspond to the support of $\gamma_k$ are zero-valued so that only, at most, $\iota_k = |\operatorname{supp}(\gamma_k)| = \mathcal{O}(1)$ rows of $\mathbf{G}^{C_k}$ are non-zero.
\end{proof}

\begin{lemma}
\label{lem: column sparsity of G^C}
    Under conditions in Theorem~1, the number of different values in each row of $\mathbf{G}^{C_k}$ is constant.
\end{lemma}
\begin{proof}
    With condition $\mathbf{\mathcal{A}5}$, the value of $g_j$ is only relevant to $\nu_j=\mathcal{O}(1)$ measurement results. So, the number of different values among the entries in the $j$th row of $\mathbf{G}^{C_k}$ is at most $2^{\nu_j}=\mathcal{O}(1)$ for all $j$.
\end{proof}

\begin{lemma}
\label{lem: high-dimensional inner product}
    For two vectors $\mathbf{a},\mathbf{b}\in \mathbb{R}^D$, such that $\|\mathbf{a}\| = \|\mathbf{b}\| = 1$ uniformly sampled at random, and for any $\eta\leq 0.1$, the probability that the inner product of these two vector larger than $\eta$ is lower bounded, i.e. $\mathbb{P}(\langle \mathbf{a}, \mathbf{b} \rangle \geq \eta) \geq \Omega(e^{-(D-1)\eta} - 0.2^{\frac{D-1}{2}})$.
\end{lemma}

\begin{proof}
    Consider a $D$-dimensional unit sphere $\mathcal{B}^D\in\mathbb{R}^D$ with its center at the origin. We can restate this lamma as sampling two vectors $\mathbf{a}$ and $\mathbf{b}$ on $\mathcal{B}^D$ uniformly at random, and the probability that their inner product is larger than $\eta$ is lower bounded by $\Omega(e^{-(D-1)\eta} - 0.2^{\frac{D-1}{2}})$. Without loss of generality, we can set $\mathbf{a} = \mathbf{e_1}$ which is a unit vector with only the first dimension non-zero, and we call its endpoint the north pole of $\mathcal{B}^D$.

    Since the volume of a $D$ dimensional sphere with radius $R$ is
    \begin{equation}
        V_D(R) = \frac{\pi^{D / 2}R^D}{\Gamma(1+D / 2)},
    \end{equation}
    and the volume of $\mathcal{B}^D$ is
    \begin{equation}
        V(\mathcal{B}^D) = \frac{\pi^{D / 2}}{\Gamma(1+D/2)}.
    \end{equation}
    Let the $D-1$ dimensional sphere $\mathcal{E}_{\mathcal{B}^D}\subseteq\mathbb{R}^D$ with $x_1=0$ be the equator of $\mathcal{B}^D$, and let $\mathcal{B}^D_\eta$ be the area such that
    \begin{equation}
    \begin{split}
        &\mathcal{B}^D_\eta \subseteq \mathcal{B}^D,\\
        &\forall p\in \mathcal{B}^D_\eta, q\in \mathcal{E}_{\mathcal{B}^D}, \min \mathfrak{d}(p,q) \geq \eta
    \end{split}
    \end{equation}
    where $\mathfrak{d}(p,q)$ is the Euclidean distance between two points in $\mathbb{R}^D$, and $\eta \leq 0.1$. Note that this guarantees that $\forall p\in \mathcal{B}^D_\eta$, the coefficient corresponding to the first dimension is at least $\eta$. Therefore, the probability of the inner product of $\mathbf{a}$ and $\mathbf{b}$ larger than $\eta$ is the ratio of the volume of $\mathcal{B}^D_\eta$ and $\mathcal{B}^D$.

    More explicitly, the volume of $\mathcal{B}^D_\eta$ is
    \begin{equation}
    \begin{split}
        V(\mathcal{B}^D_\eta) &= 2\int_\eta^1 V_{D-1}(\sqrt{1-h^2})dh\\
        &= \frac{2\pi^{\frac{D-1}{2}}}{\Gamma(1+\frac{D-1}{2})} \int_\eta^1 (1-h^2)^{\frac{D-1}{2}}dh.\\
    \end{split}
    \end{equation}
    To lower bound $V(\mathcal{B}^D_\eta)$, we introduce $h^*\in (0.9,1)$ such that $1-(h^*)^2 = e^{-2h^*}$. This guarantees that
    \begin{equation}
        1-h^2 \geq e^{-2h}, \forall h\in(0,h^*).
    \end{equation}
    So,
    \begin{equation}
        \begin{split}
            V(\mathcal{B}^D_\eta) &= \frac{2\pi^{\frac{D-1}{2}}}{\Gamma(1+\frac{D-1}{2})} \int_\eta^1 (1-h^2)^{\frac{D-1}{2}}dh\\
            &\geq \frac{2\pi^{\frac{D-1}{2}}}{\Gamma(1+\frac{D-1}{2})} \int_\eta^{h^*} (1-h^2)^{\frac{D-1}{2}}dh\\
            &\geq \frac{2\pi^{\frac{D-1}{2}}}{\Gamma(1+\frac{D-1}{2})} \int_\eta^{h^*} e^{-2h\cdot\frac{D-1}{2}}dh\\
            &=\frac{2\pi^{\frac{D-1}{2}}}{\Gamma(1+\frac{D-1}{2})} \int_\eta^{h^*} e^{-(D-1)h}dh\\
            &= \frac{2\pi^{\frac{D-1}{2}}}{\Gamma(1+\frac{D-1}{2})}\cdot \frac{1}{D-1}\cdot(e^{-(D-1)\eta} - e^{-(D-1)h^*})\\
            &= \frac{\pi^{\frac{D}{2}}}{\Gamma(1+\frac{D-1}{2})\cdot\frac{D-1}{2}}\cdot\frac{1}{\sqrt{\pi}}\cdot(e^{-(D-1)\eta} - (1-(h^*)^2)^{\frac{D-1}{2}})\\
            &\geq \frac{\pi^{\frac{D}{2}}}{\Gamma(1+\frac{D-1}{2})\cdot\frac{D+1}{2}}\cdot\frac{1}{\sqrt{\pi}}\cdot(e^{-(D-1)\eta} - (e^{-2h^*})^{\frac{D-1}{2}})\\
            &\geq \frac{\pi^{\frac{D}{2}}}{\Gamma(1+\frac{D}{2})}\cdot\frac{1}{\sqrt{\pi}}\cdot(e^{-(D-1)\eta} - (1-(h^*)^2)^{\frac{D-1}{2}})\\
            &= V(\mathcal{B}^D)\cdot\frac{1}{\sqrt{\pi}}(e^{-(D-1)\eta} - (1-(h^*)^2)^{\frac{D-1}{2}})\\
            &\geq V(\mathcal{B}^D)\cdot\frac{1}{\sqrt{\pi}}(e^{-(D-1)\eta} - 0.2^{\frac{D-1}{2}}),
        \end{split}
    \end{equation}
    where the last inequality is obtained by $h^*\geq 0.9$ and $(1-(h^*)^2)\leq 0.2$. Thus, we have proved that:
    \begin{equation}
        \frac{V(\mathcal{B}^D_\eta)}{V(\mathcal{B}^D)} \geq \frac{1}{\sqrt{\pi}}(e^{-(D-1)\eta} - 0.2^{\frac{D-1}{2}}).
    \end{equation}
    Since $\eta \leq 0.1$, the lower bound is larger than $0$ when $D > 1$, so it is not trivial. Thus, we have proved that
    \begin{equation}
        \mathbb{P}(\langle \mathbf{a}, \mathbf{b} \rangle \geq \eta) \geq \Omega(e^{-(D-1)\eta} - 0.2^{\frac{D-1}{2}}).
    \end{equation}
\end{proof}

\begin{remark}
    In Lemma~\ref{lem: high-dimensional inner product}, the dimension $D$ can be interpreted as the degree of freedom of the inner product of two vectors $\mathbf{a}$ and $\mathbf{b}$.
\end{remark}

Now, with these lemmas proven, we are ready to combine the above lemmas and prove that quantum gradients for variational LOCC-assisted circuits are free of barren plateaus.

\begin{proof}[Proof of Theorem~1]
    Recall Eq.~\eqref{eq:Hadamard product form of the gradients},
\begin{equation}
    \frac{\partial \Tr \left[\hat{H} \Psi_{\gamma}\right]}{\partial \gamma_k} = \langle\mathbf{G}^{C_k}, \mathbf{G}^{Q_k} \rangle_F.
\end{equation}
With condition $\mathbf{\mathcal{A}3}$ and Lemma~\ref{lemma: entries of G^Q is not exponentially small}, the entries in neither $\mathbf{G}^{C_k}$ nor $\mathbf{G}^{Q_k}$ are exponentially small as $n$ scales. We now focus on the inner product structure. The Frobenius inner product of two $l$-by-$2^m$ matrices can be interpreted as the inner product of two vectors with length $l\cdot 2^m$. However, Lemma~\ref{lem:row sparsity of G^C} indicates that only the entries corresponding to a constant number of non-zero number rows of $\mathbf{G}^{C_k}$ needs to be considered. Among these rows, Lemma~\ref{lem: column sparsity of G^C} and Lemma~\ref{lem:sparsity of G^Q} indicate that the number of different values in each row is a constant. Thus, the inner product of the two high-dimensional vectors only has a number of free parameters $\mathcal{D}$,
\begin{equation}
    \mathcal{D} \leq \mathcal{O}(\iota_k \cdot (\max_i(\zeta_i)+\max_j(2^{\nu_j})) = \mathcal{O}(1)
\end{equation}
where $\zeta_i$, $\iota_k$ and $\nu_j$ follow the definitions in Lemma~\ref{lem: column sparsity of G^C}, Lemma~\ref{lem:row sparsity of G^C}, and Lemma~\ref{lem:sparsity of G^Q}.

Consequently, with Lemma~\ref{lem: high-dimensional inner product}, where the dimension can be interpreted as the number of free parameters $\mathcal{D}$, the probability that $\frac{\partial \Tr \left[\hat{H} \Psi_{\gamma}\right]}{\partial \gamma_k}\geq \epsilon$ is lower bounded by
\begin{equation}
    \mathbb{P}(|\frac{\partial \Tr \left[\hat{H} \Psi_{\gamma}\right]}{\partial \gamma_k} \geq \epsilon)| \geq \Omega(e^{-(\mathcal{D}-1)\epsilon} - 0.2^{\frac{\mathcal{D}-1}{2}})
\end{equation}
where $\mathcal{D}=\mathcal{O}(1)$, independent of $n$. So, the probability $ \mathbb{P}(\frac{\partial \Tr \left[\hat{H} \Psi_{\gamma}\right]}{\partial \gamma_k} \geq \epsilon)$ is lower bounded by a constant independent of $n$. This means that the gradient will not decay as $n$ scales, which completes our proof of the absence of barren plateaus in LOCC-VQE.
\end{proof}

\subsection{Numerical evidence}
 We also numerically demonstrate the absence of barren plateaus through the training process of the transverse-Ising model defined in Section~\ref{appendix:tfIsing}. Quantum gradient information of the first optimization iteration in LOCC-VQE for the transverse-Ising model is recorded by the mean of the absolute values. The numerical results are shown in Fig.~\ref{fig:Gradients scaling with n}.

\begin{figure}[htbp!]
    \centering
    \includegraphics[width=0.45\linewidth]{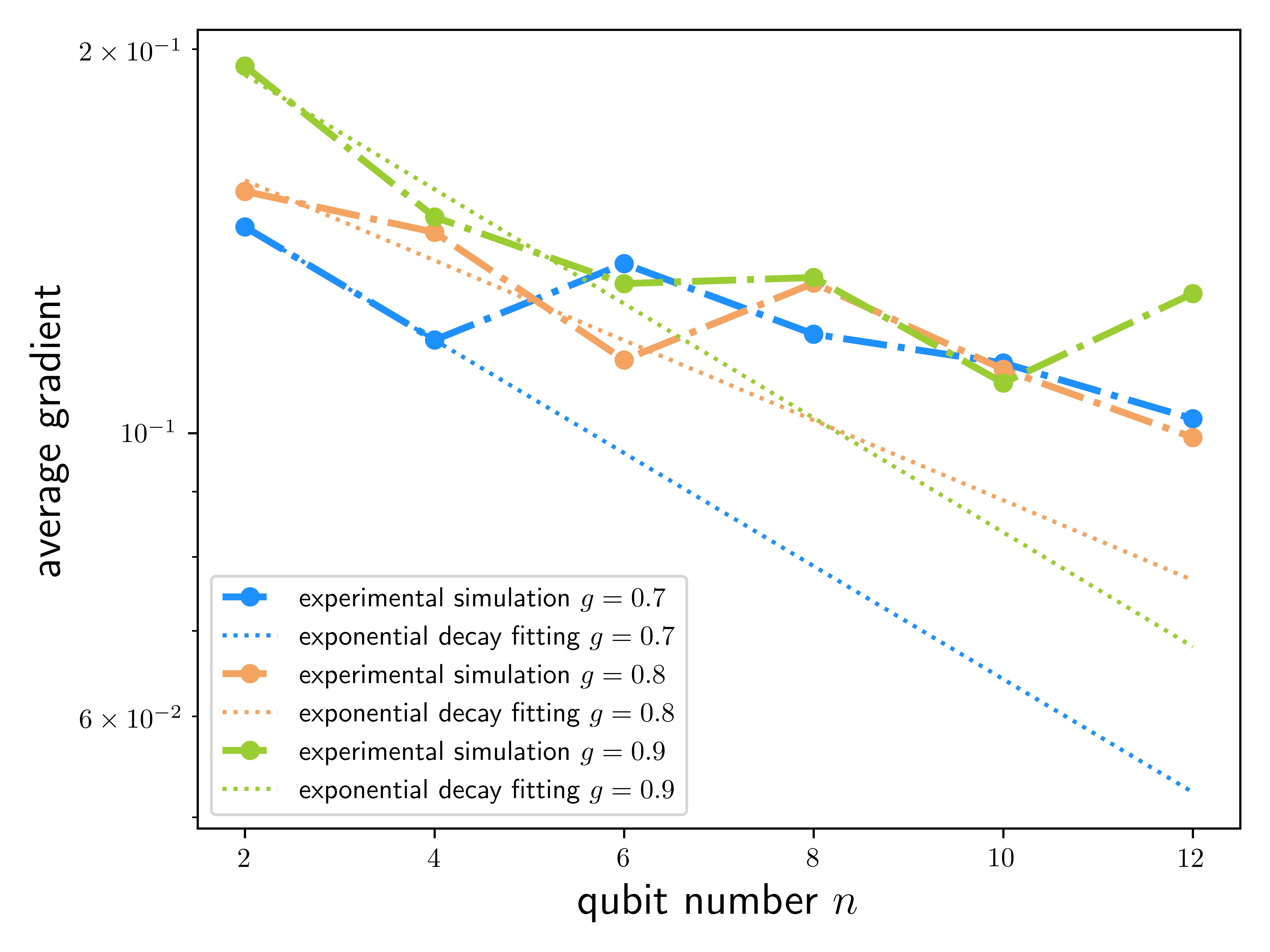}
    \caption{Average gradients scaling with $n$. The mean of the absolute value of the gradients of the first optimization iteration for $n\in[2,12]$. The dot lines represent the exponential decay fitted from the first few average gradient values. The blue, orange, and green curves represent different coupling coefficients $g=0.7,0.8,0.9$, respectively.}
    \label{fig:Gradients scaling with n}
\end{figure}

Notice that under conditions in Theorem~1, the gradients of variational LOCC-assisted circuit do not vanish as the number of qubits $n$ scales in the asymptotic limit. Due to computational limitations, we can only numerically implement smaller system sizes, where the nonvanishing property of the gradients is less obvious. This is because, from the light cone perspective, the boundaries of the system have a stronger effect on the propagation of information in smaller systems,  leading to the decay of the average gradient as shown in Fig.~\ref{fig:Gradients scaling with n}. However, even with a small system size, Fig.~\ref{fig:Gradients scaling with n} demonstrates that the quantum gradients for variational LOCC-assisted circuits have a significant separation remain significantly separated from exponentially decaying values under conditions in Theorem~1.

\section{Circuit ansatzes and training details}
\label{app:arch}

\subsection{Architecture of LOCC-assisted circuits}
\label{sec:Architecture of LOCC-assisted circuits}

In our simulation, we use circuits with tunable $\mathcal{SU}(4)$ \cite{khaneja_cartan_2000, earp_constructive_2005} in our circuit architecture. We use the Cartan decomposition of $\mathcal{SU}(4)$ to parameterize the unitary local two-qubit gates in our optimization process, as illustrated in Fig.~\ref{fig:Cartan decompostion}. $R_x,R_y,R_z$ are single-qubit rotation gates with generators $X,Y,Z$. $R_{xx},R_{yy},R_{zz}$ are two-qubit rotation gates with generators $X\otimes X,Y\otimes Y,Z\otimes Z$. Parametrized by $\theta_x,\theta_y,\theta_z,\theta_{xx},\theta_{yy},\theta_{zz}\in[0,2\pi)$, $R_x(\theta_x)R_y(\theta_y)R_z(\theta_z)$ forms an universal single qubit gate. Together with $R_{xx}(\theta_{xx})R_{yy}(\theta_{yy})R_{zz}(\theta_{zz})$, this block can represent a universal two-qubit unitary gate.

\begin{figure}[htbp!]
    \centering
    \includegraphics[width=0.7\linewidth]{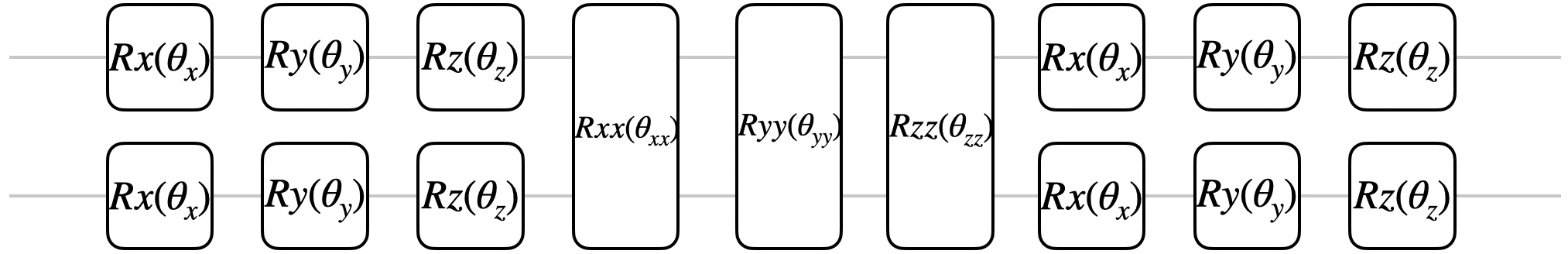}
    \caption{Cartan decomposition of unitary two-qubit gate.}
    \label{fig:Cartan decompostion}
\end{figure}

To prepare the GHZ state with perturbations and the ground states of the transverse-field Ising model, we adopt a similar architecture for LOCC-VQE. We first apply two layers of local two unitary gates parametrized through Cartan decomposition. Each local two-qubit gate acts on a data qubit and its neighboring ancillary qubit. We then measure all ancillary qubits and feed the measurement outputs into the LOCC protocol function $g$. We set the output of $g$ to be the parameters of one layer of the single-qubit rotation gate on all data qubits. We empirically set the LOCC protocol function $g$ to be the summation of one layer of neural network and neurons linking measurement results from far-apart ancillary qubits to enhance the power of preparing long-range entanglements. The empirical design of $g$ also shares the idea of the LOCC preparation protocol of the unperturbed GHZ state in \cite{piroli_quantum_2021}.

To prepare the ground state of the perturbed rotated surface code, we adopt a similar architecture of error correction. We first applied four layers of parametrized local two-qubit gates through the Cartan decomposition. Each two-qubit gate acts on a data qubit and an associated syndrome qubit following the stabilizer formalism of rotated surface code. Later, all ancillary qubits are measured, and the results are fed into the LOCC protocol function $g$. The output of $g$ will be set as the parameters of one layer of the single-qubit rotation gate on all data qubits at the end of the quantum circuit.

\subsection{Architecture of unitary circuits}
\label{appendix:brick circuit}
For all unitary VQE used in this paper, we adopted the parametrized brick wall quantum circuit as the circuit ansatz, in which all two-qubit gates are parametrized through the Cartan decomposition. The 1D brick wall quantum circuit has a structure as illustrated in Fig.~\ref{fig:brick wall quantum circuit}. It is formed by consecutive layers of interleaving parametrized local two-qubit gates with a compact layout.

\begin{figure}[htbp!]
    \centering
    \includegraphics[width=0.2\linewidth]{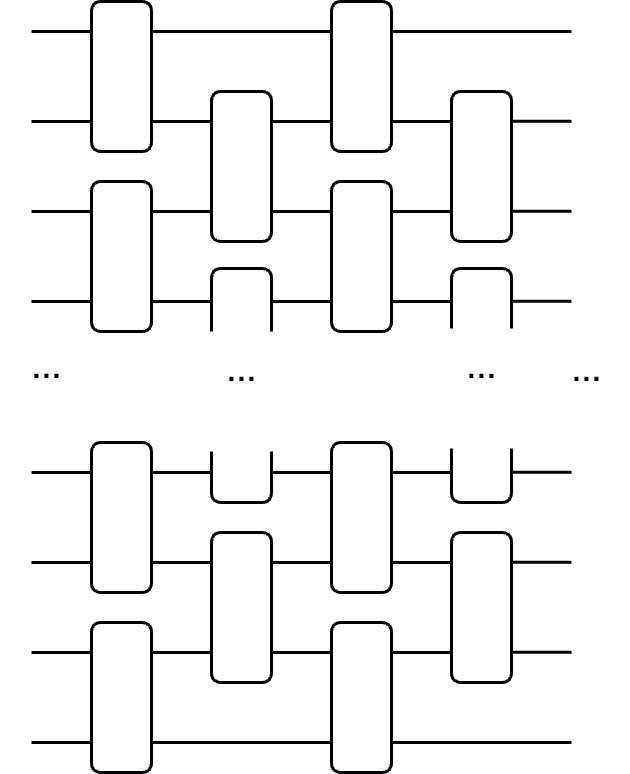}
    \caption{Brick-wall quantum circuit.}
    \label{fig:brick wall quantum circuit}
\end{figure}

For the 2D version, the brick wall quantum circuit is analogous to the 1D version, consisting of consecutive layers of interleaving parametrized near-neighbor two-qubit gates with a compact layout.

\subsection{Training details}

The numerical simulation is computationally demanding due to the exponential scaling of the computational time and memory as the qubit number $n$ scales up. We use the current state-of-the-art tensor-network-based quantum simulation technique \cite{zhang_tensorcircuit_2023} combined with an efficient machine learning framework and massive parallelization during the simulation. With efforts combined, we can simulate up to 20 qubits implementing LOCC-VQE. For parameter optimization, we utilize
the Adam optimizer with the learning rate set to be $0.01$.

To save numerical costs, we assume a sufficient number of samplings in most of our numerical simulations. We use block-diagonal unitary two-qubit gates to represent classical control and single-qubit rotations of the data qubits. We also use the sample-based version to prepare the ground state for a smaller system size. In each training iteration, we set the sample round to 100 and used a combination of one layer of neural network as well as empirically set functions to enhance the representability of the LOCC protocol function. We illustrate the results in Section~\ref{appendix:sample_based}.

Due to numerical costs, we didn't simulate a larger physical system or use more sample rounds in the sample-based simulation. However, since most computation resources are used for parallel sampling and quantum circuit simulation, real quantum experiments would not face this problem. The sampling cost of LOCC-VQE is the same as the unitary VQE protocol as discussed in the main text.

We summarized the number of iterations for circuit-parameter optimization in Table~\ref{tab:training_iteration}.

\begin{table}[h!]
\caption{\label{tab:training_iteration}
Numbers of iterations for circuit-parameter optimization.\footnote{To have fair comparisons, for each model, we give the experiments using LOCC-VQE and unitary VQE the same amount of iterations for each data point, and we also make sure that the number of iterations is enough for the optimization of unitary VQE to converge.}
}
\begin{tabular*}{0.9\textwidth}{@{\extracolsep{\fill}} c c c c }
\toprule

\multirow{2}{*}{\textrm{Model}}& \textrm{LOCC-VQE}
& \textrm{LOCC-VQE} & \multirow{2}{*}{\textrm{Unitary VQE}}\\
 &\textrm{(Sufficient samplings)} & \textrm{(Finite-shot samplings)} & \\
\colrule

8-qubit perturbed GHZ & \multirow{2}{*}{2000} & \multirow{2}{*}{-} & \multirow{2}{*}{2000} \\
($X$ perturbation) & & & \\
\colrule

8-qubit perturbed GHZ & \multirow{2}{*}{2000} & \multirow{2}{*}{-} & \multirow{2}{*}{2000} \\
($Y$ perturbation) & & & \\
\colrule

8-qubit perturbed GHZ &\multirow{2}{*}{2000} & \multirow{2}{*}{-} & \multirow{2}{*}{2000} \\
($Z$ perturbation) & & & \\
\colrule

4-qubit perturbed GHZ & \multirow{2}{*}{-} & \multirow{2}{*}{2000} & \multirow{2}{*}{2000} \\
($X$ perturbation) & & & \\
\colrule

4-qubit perturbed GHZ & \multirow{2}{*}{-} & \multirow{2}{*}{2000} & \multirow{2}{*}{2000} \\
($Y$ perturbation) & & & \\
\colrule

4-qubit perturbed GHZ & \multirow{2}{*}{-} & \multirow{2}{*}{2000} & \multirow{2}{*}{2000} \\
($Z$ perturbation) & & & \\
\colrule

Perturbed rotated surface code &\multirow{2}{*}{2000} & \multirow{2}{*}{-} & \multirow{2}{*}{2000} \\
($X$($Z$) perturbation) & & & \\
\colrule

Perturbed rotated surface code &\multirow{2}{*}{2000} & \multirow{2}{*}{-} & \multirow{2}{*}{2000} \\
($Y$ perturbation) & & & \\

\colrule
\multirow{2}{*}{8-qubit transverse-field Ising} & \multirow{2}{*}{2000}  & \multirow{2}{*}{-}  & \multirow{2}{*}{2000} \\
 & & & \\
 \colrule

\multirow{2}{*}{4-qubit transverse-field Ising} & \multirow{2}{*}{-}  & \multirow{2}{*}{2000}  & \multirow{2}{*}{2000} \\
 & & & \\
\colrule

Perturbed toric code & \multirow{2}{*}{6000} & \multirow{2}{*}{-} & \multirow{2}{*}{6000} \\
($X$($Z$) perturbation) & & & \\
\colrule

Perturbed toric code & \multirow{2}{*}{-} & \multirow{2}{*}{10000} & \multirow{2}{*}{10000} \\
($Y$ perturbation) & & & \\
\colrule

\botrule
\end{tabular*}
\end{table}

\subsection{Structures of LOCC protocols}

\begin{figure}[ht!]
    \centering
    \includegraphics[width=0.75\linewidth]{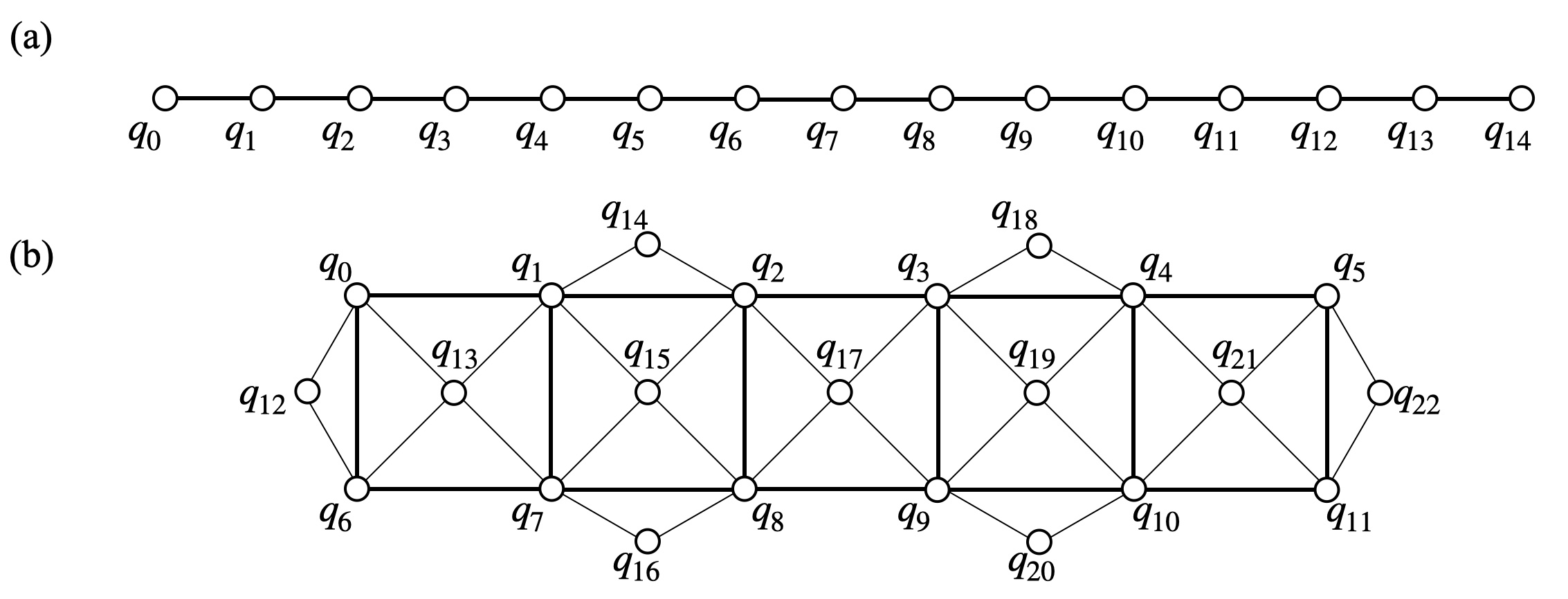}
    \caption{Indices of qubits. \textbf{(a)} 1D chain model. \textbf{(b)} Rotated surface code model.}
    \label{fig:qubit_indexing}
\end{figure}
We index the qubits following Fig.~\ref{fig:qubit_indexing} for both the 1D chain model and the rotated surface code model.

In the following, we give how we design the structures of the LOCC protocol functions $g$ in numerical experiments with unlimited shots: the perturbed GHZ state in Table~\ref{tab:g_function_for_GHZ}, the transverse-field Ising model in Table~\ref{tab:g_function_for_TF}, and the perturbed rotated surface code in Table~\ref{tab:g_function_for_PRSC}.

\begin{table}[ht!]
\caption{\label{tab:g_function_for_GHZ}
Structure of the LOCC protocol function in LOCC-VQE for perturbed GHZ state preparation.
}
\begin{tabular*}{0.75\textwidth}{@{\extracolsep{\fill}} c c c c }
\toprule
\textrm{Layer}&
\textrm{Indices of Control Bits\footnote{The indices used for bits are in correspondence with the indices of qubits where measurement are performed on.}}&
\textrm{Classical Logical Operations}&
\textrm{Indices of Target Qubits\footnote{For target qubits we mean the set of qubits where single-qubit unitary gates act on, the parameters of these gates, as defined in Section~\ref{sec:Architecture of LOCC-assisted circuits} are variationally determined based on the values of the control bits.}}\\
\colrule
1 & 9 & - & 1\\
2 & 10 & - & 3\\
3 & 12 & - & 5\\
4 & 14 & - & 7\\
5 & - & $8 \leftarrow 8\oplus9$ \footnote{`$a\oplus b$' stand for the classical XOR operation on the values in the classical register a and b, and $c\leftarrow d$ represents replacing the value in the classical register c with the value in d.} & -\\
6 & - &$10 \leftarrow 10\oplus11$ & -\\
7 & - &$12 \leftarrow 12\oplus13$ & -\\
8 & 8 & - & 2\\
9 & 12 &- & 6\\
10 & 10 &- & 3\\
11 & - &$8 \leftarrow 8\oplus10$ & -\\
12 & - &$12 \leftarrow 12\oplus14$ & -\\
13 & 8 &- & range[4:8]\footnote{The term ``rang[a:b]'', standing for the set $\{a,a+1,\ldots,b-1\}$.}\footnote{We use a set of indices to simplify the representation of each (qu)bit indexed by the element in this set of indices as the control bit/target qubit independently.}\\
14 &-&$8 \leftarrow 8\oplus12$&-\\
15 & range[8:15]&-&range[0:8]\\
\botrule
\end{tabular*}
\end{table}

\begin{table}[ht!]
\caption{\label{tab:g_function_for_TF}
Structure of the LOCC protocol function in LOCC-VQE for the transverse-field Ising model's ground state preparation.\footnote{The notations used in this table follow the definitions in Table~\ref{tab:g_function_for_GHZ}.}
}
\begin{tabular*}{0.75\textwidth}{@{\extracolsep{\fill}} c c c c }
\toprule
\textrm{Layer}&
\textrm{Indices of Control Bits}&
\textrm{Classical Logical Operations}&
\textrm{Indices of Target Qubits}\\
\colrule
1 & range[8:15]&-&range[0:8]\\
2 & 9 & - & 1\\
3 & 10 & - & 3\\
4 & 12 & - & 5\\
5 & 14 & - & 7\\
6 & - & $8 \leftarrow 8\oplus9$  & -\\
7 & - &$10 \leftarrow 10\oplus11$ & -\\
8 & - &$12 \leftarrow 12\oplus13$ & -\\
9 & 8 & - & 2\\
10 & 12 &- & 6\\
11 & 10 &- & 3\\
12 & - &$8 \leftarrow 8\oplus10$ & -\\
13 & - &$12 \leftarrow 12\oplus14$ & -\\
14 & 8 &- & range[4:8]\\
15 &-&$8 \leftarrow 8\oplus12$&-\\
\botrule
\end{tabular*}
\end{table}

\begin{table}
\centering
\caption{\label{tab:g_function_for_PRSC}
Structure of the LOCC protocol function in LOCC-VQE for the perturbed rotated surface code state preparation.\footnote{The notations used in this table follow the definitions in Table~\ref{tab:g_function_for_GHZ}.}
}
\begin{tabular*}{0.75\textwidth}{@{\extracolsep{\fill}} c c c c }
\toprule
\textrm{Layer}&
\textrm{Indices of Control Bits}&
\textrm{Classical Logical Operations}&
\textrm{Indices of Target Qubits}\\
\colrule
1 & range[12:19]&-&range[0:12]\\
2 & - & $18 \leftarrow 12\oplus18$ & -\\
3 & - & $17 \leftarrow 13\oplus17$ & -\\
4 & - & $16 \leftarrow 14\oplus16$ & -\\
5 & 18 & - & range[0:12]\\
6 & 17 & -  & range[0:12]\\
7 & 16 & - & range[0:12]\\
8 & - &$14 \leftarrow 13\oplus14$ & -\\
9 & - & $17 \leftarrow 16\oplus17$ & -\\
10 & 13 &- & range[0:12]\\
11 & 16 &- & range[0:12]\\
12 & - &$15 \leftarrow 12\oplus15\oplus18$ & -\\
13 & 15 &- & range[0:12]\\
14 & - &$15 \leftarrow 13\oplus15\oplus16$ & -\\
15 & - &$12 \leftarrow 12\oplus14\oplus17$ & -\\
16 & 12 &- &range[0:12]\\
17 & 15 &- &range[0:12]\\
18 & - &$18 \leftarrow 14\oplus16\oplus18$ & -\\
19 & 18 &- & range[0:12]\\
20 & - &$15 \leftarrow 14\oplus15\oplus16$ & -\\
21 & 15 &- & range[0:12]\\

\botrule
\end{tabular*}
\end{table}

\section{Supplementary numerical results}
\label{app:suppnum}

\subsection{Greenberger–Horne–Zeilinger state with various perturbations}
\label{sec:GHZ_more_results}
Besides the Pauli $X$ perturbation shown in Fig.~3 in the main text, we also numerically simulate the GHZ state with Pauli $Y$ and Pauli $Z$ perturbations. The results are illustrated in Fig.~\ref{fig:GHZ_8_YandZ_EandQMI}, LOCC-VQE achieves $10^{-3}$ relative error in ground state energy while unitary VQE can only achieve $10^{-1}$ relative error.

\begin{figure}[htbp!]
    \centering
    \includegraphics[width=0.75\linewidth]{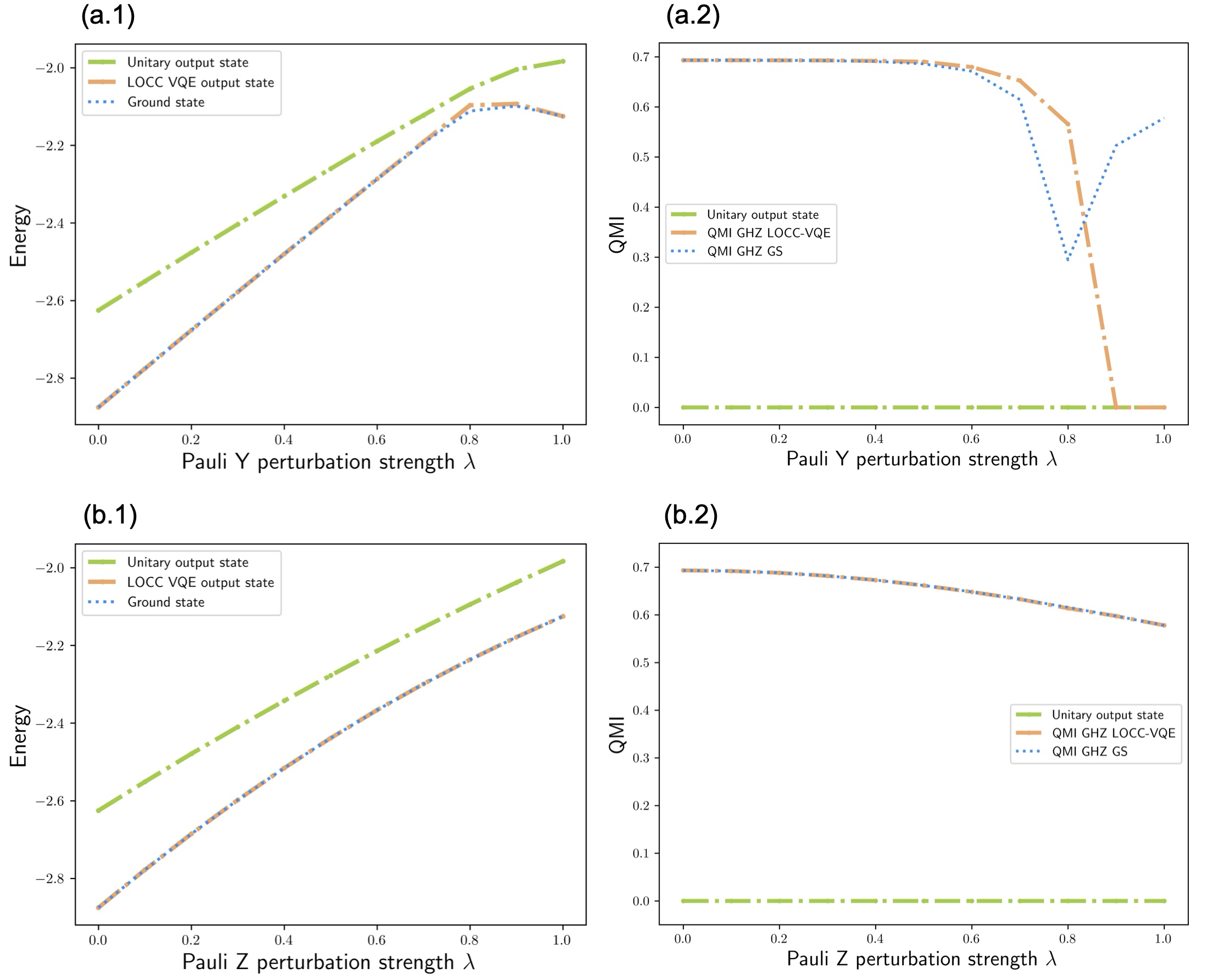}
    \caption{Numerical simulation results of solving the parent Hamiltonian of the $8$-qubit GHZ state with Pauli $Y$ and $Z$ perturbations with depth two circuits. \textbf{(a)} Comparison between the energy optimization results through LOCC-VQE and unitary VQE with depth two circuits. \textbf{(b)} Comparison between the QMI between subsystems. subsystem $A$ and $C$ as shown in Fig.~2~(a).}
    \label{fig:GHZ_8_YandZ_EandQMI}
\end{figure}
\subsection{Transverse-field Ising model}
\label{appendix:tfIsing}

As another case of perturbated GHZ states, we also test the 1D transverse-field Ising model,
\begin{equation}
    \hat{H}_{\text{tfIsing}}=-\sum_{\langle i, j\rangle} Z_i Z_j - \lambda \sum_j X_j.
    \label{eq: tfIsing}
\end{equation}
The difference is that we do not introduce the degeneracy-breaking term, namely setting $h=0$ in Eq.~(5). The model exhibits long-range entanglement at its critical point $\abs{\lambda}=1$ with a sufficiently large system size.

We numerically compare the energy accuracy of LOCC-VQE in solving a $8$-qubit transverse-field Ising model to its unitary counterpart, with results shown in Fig.~\ref{fig:Energy Ising 8}. The result suggests the advantages of LOCC-VQE when $\lambda$ is near or larger than $1$, the phase transition point; and when $\lambda < 1$, LOCC-VQE still exhibits comparable accuracy with unitary VQE. This result suggests the general applicability of our approach.

\begin{figure}[htbp!]
    \centering
    \includegraphics[width=0.75\linewidth]{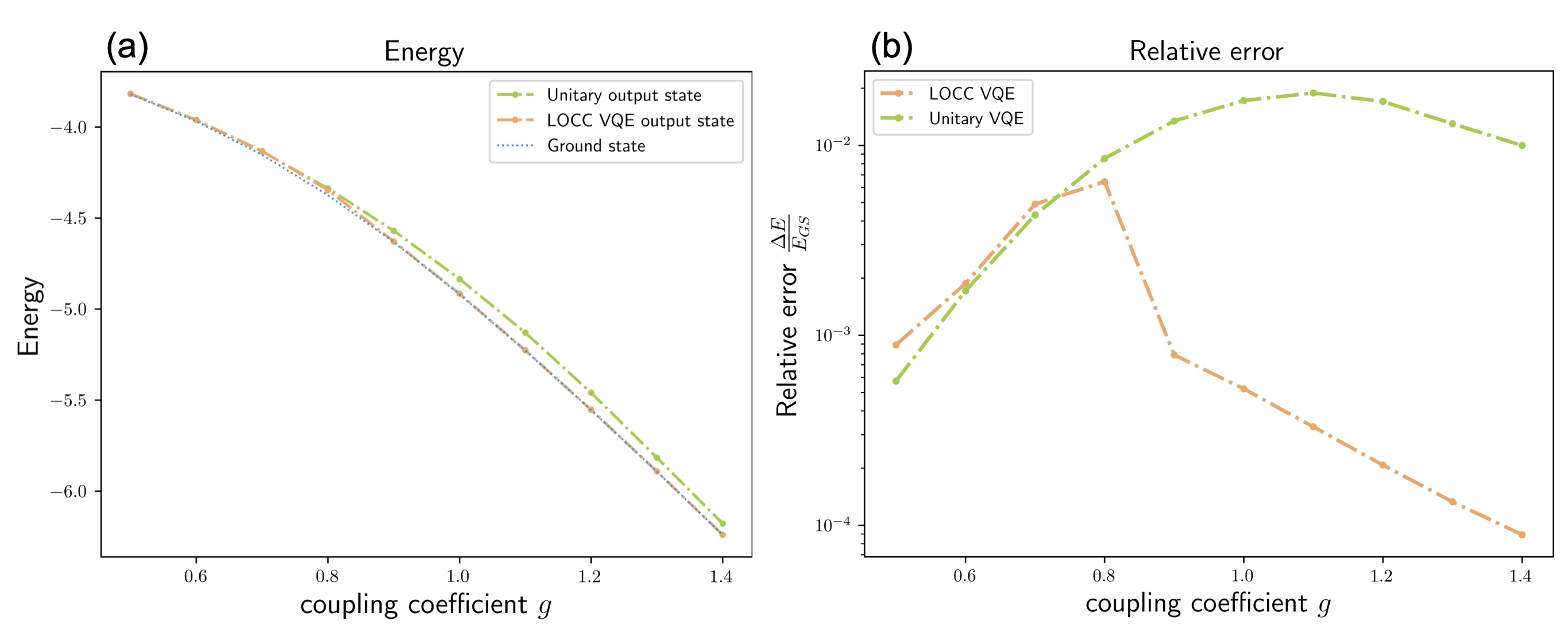}
    \caption{Numerical simulation results of solving the $8$-qubit transverse-field Ising model. \textbf{(a)} Comparison between the energy optimization results through LOCC-VQE and unitary VQE with depth two circuits. \textbf{(b)} Comparison between the relative error of ground state energy optimization results, $\frac{\Delta E}{E_{GS}}=\frac{E-E_{GS}}{E_{GS}}$, through LOCC-VQE and unitary VQE with depth two circuits.}

    \label{fig:Energy Ising 8}
\end{figure}

\subsection{Toric code compared with measurement-based variational quantum eigensolver}

The toric code is a quantum error-correcting code defined on a two-dimensional rectangular lattice with periodic boundary conditions. The ground states of the toric code Hamiltonian possess long-range entanglement, enabling the storage of logical information. In previous work \cite{ferguson_measurement-based_2021}, this model, under perturbations, is used to test measurement-based variational quantum eigensolver (MB-VQE). Here, we use it to also test our approach, which exhibits a more accurate result.

\begin{figure}[htbp!]
    \centering
    \includegraphics[width=0.5\linewidth]{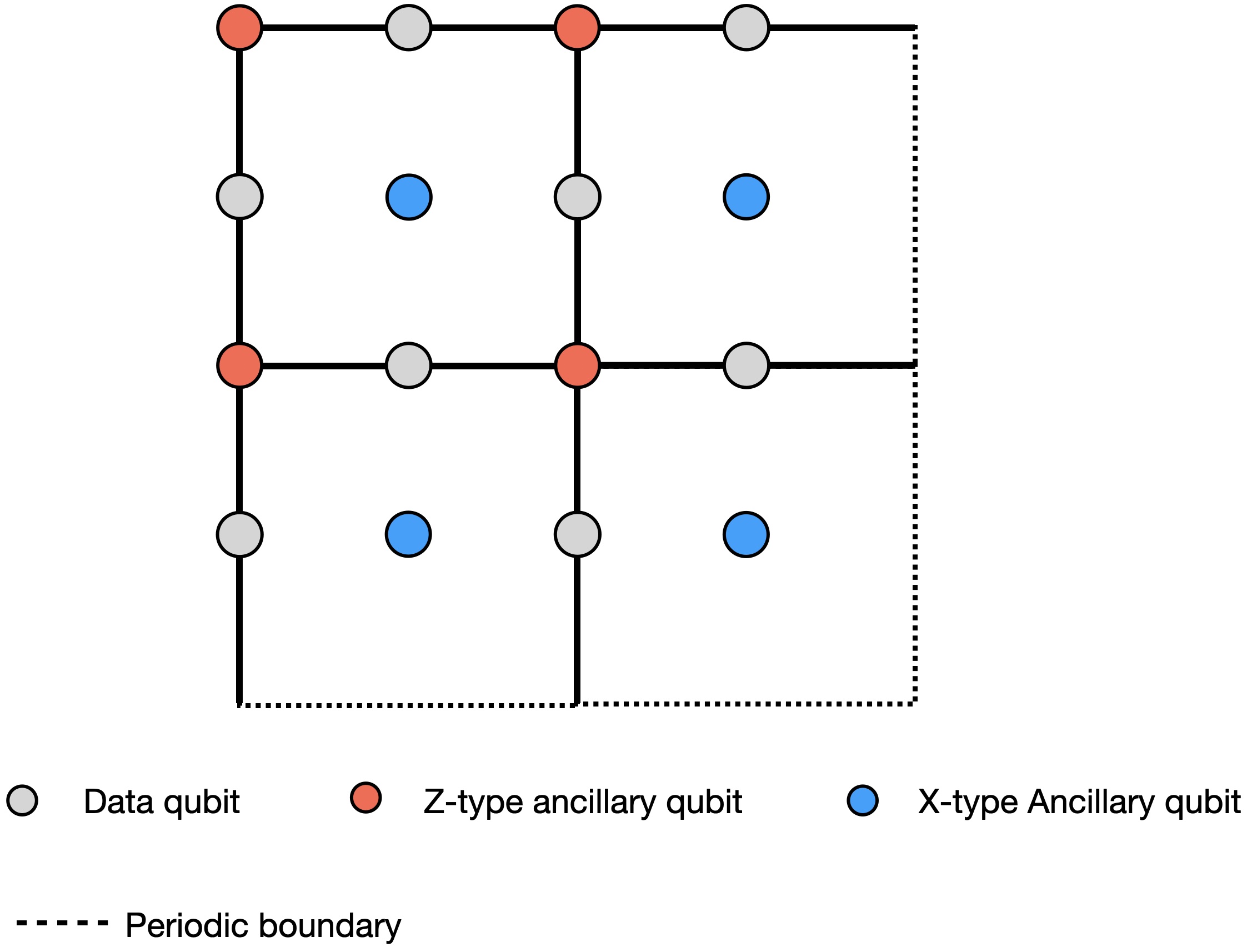}
    \caption{Two-dimensional toric code's lattice. An edge represents a data qubit, a vertex represents a Z-type ancillary qubit corresponding to a Z-stabilizer, and a plaque represents an X-type ancillary qubit corresponding to an X-stabilizer. The dotted lines on the boundary represent periodic boundary conditions, i.e., the left-most and the right-most edges are equivalent, and the upper-most and the lower-most edges are equivalent.}
    \label{fig:Toric code}
\end{figure}

In our numerical tests, we add a magnetic field in the $Z$ direction as a perturbation, resulting in the following Hamiltonian
\begin{equation}
\label{eq: H_toric}
    \hat{H}_{\text{tor}}(\lambda) = -(1-\lambda)\sum_v A_v - (1-\lambda)\sum_p B_p - \lambda \sum_{i=1}^{N_xN_y} P_i,
\end{equation}
where $N_x$ and $N_y$ are width and height of the regular lattice, $A_v$ and $B_p$ are stabilizers for the unperturbed rotated surface code, and $\lambda$ is the perturbation strength and $P_i \in X_i, Y_i, Z_i$ on the $i$-th site. The $Z$-type stabilizers $A_v$ and $X$-type stabilizers $B_p$ correspond to vertices and plaques on the lattice respectively, illustrated in Fig.~\ref{fig:Toric code}.

\begin{figure}[htbp!]
    \centering
    \includegraphics[width=0.75\linewidth]{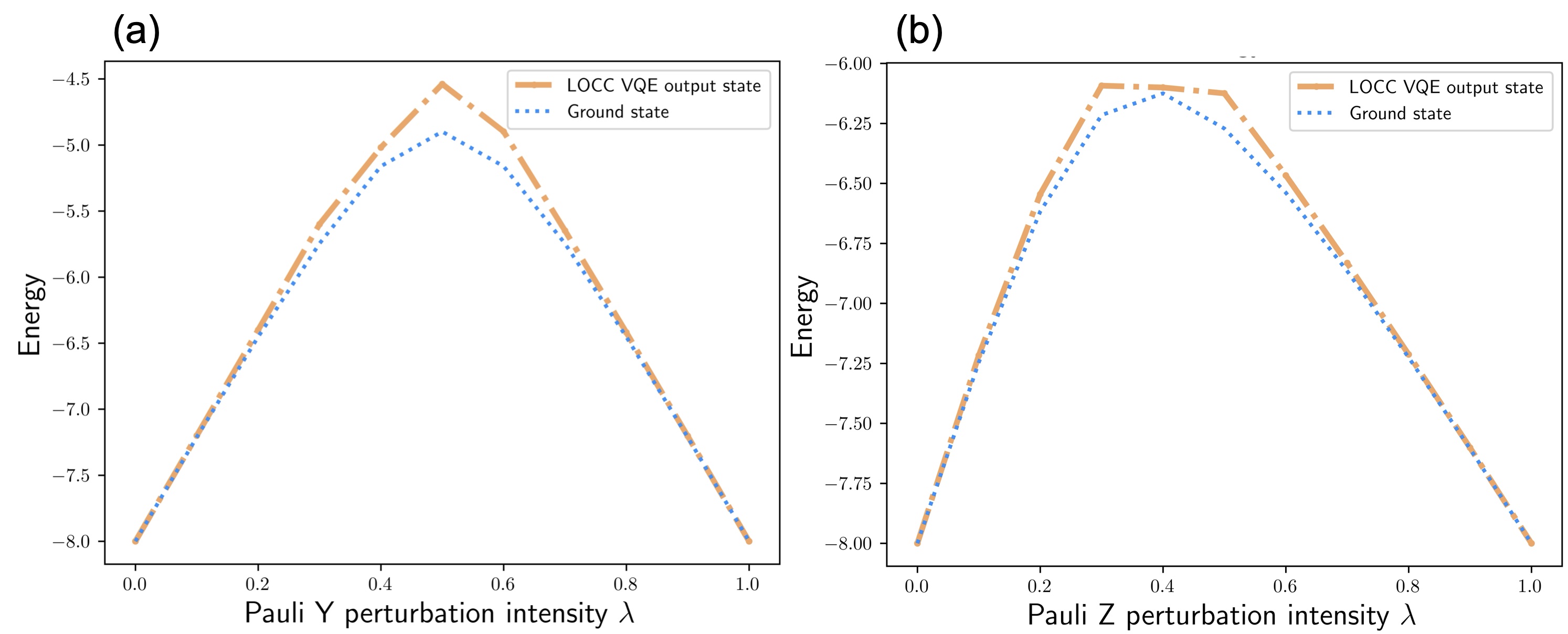}

    \caption{Numerical simulation results of the toric code with Pauli $Y$ \textbf{(a)} perturbation and Pauli $X$(Z) \textbf{(b)} perturbation. Due to symmetry, Pauli $X$ perturbation is equivalent to Pauli $Z$ perturbation up to a change of basis.}
    \label{fig:Energy_toric_code}
\end{figure}

The numerical results are illustrated in Fig.~\ref{fig:Energy_toric_code}. We can achieve $10^{-2}$ accuracy in relative error among all perturbation strengths. We compare the results of LOCC-VQE with the result of MB-VQE \cite{ferguson_measurement-based_2021} over the same perturbed toric code model. LOCC-VQE can achieve higher accuracy than MB-VQE, as illustrated in Table~\ref{tab:compare with MBVQE}.

\begin{table}[htbp!]
\caption{\label{tab:compare with MBVQE}
Comparison between the relative error of preparing the ground states of perturbed toric code by LOCC-VQE and MB-VQE.
}
\begin{tabular*}{0.75\textwidth}{@{\extracolsep{\fill}} c c c }
\toprule
\textrm{$\max\limits_{\lambda}(\frac{\Delta E}{E_{GS}})$\footnote{The form of perturbation used in MB-VQE is slightly different from the form in Eq.~\eqref{eq: H_toric}. However, we can still compare the results by calculating the relative error.}}&
\textrm{Pauli $Y$ perturbation}&
\textrm{Pauli $Z$(X) perturbation }\\
\colrule
\textrm{LOCC-VQE} & 0.0735 & 0.0234 \\
\textrm{MB-VQE \cite{ferguson_measurement-based_2021}} & -\footnote{Pauli $Y$ perturbation is not demonstrated in this reference.} & 0.0378 \\
\botrule
\end{tabular*}

\end{table}

\subsection{Finite-shot simulation}
\label{appendix:sample_based}
We numerically simulate the sampling process during midcircuit measurements, and we call it sampling-based LOCC-VQE in our numerical simulations. The results of using sampling-based LOCC-VQE for the 4-qubit GHZ states with perturbations are shown in Fig.~\ref{fig:GHZML4}, and the results for transverse-field Ising model are depicted in Fig.~\ref{fig:IsingML4}. As shown in the numerical results, the states prepared by LOCC-VQE have $10^{-2}$ relative error for preparing GHZ state with Puali $X$ and $Y$ perturbations, $10^{-4}$ relative error for preparing GHZ state with Puali $Z$ perturbation, and $10^{-2}$ relative error for preparing ground state of the transverse-field Ising model, even with the presence of sample inaccuracy (40 rounds of samples(shots) for each measurement in the experiments for perturbed GHZ state preparation and 1000 rounds of samples(shots) for each measurement in the experiments for preparing ground state of the transverse-field Ising model).

\begin{figure}[htbp!]
    \centering
    \includegraphics[width=0.9\linewidth]{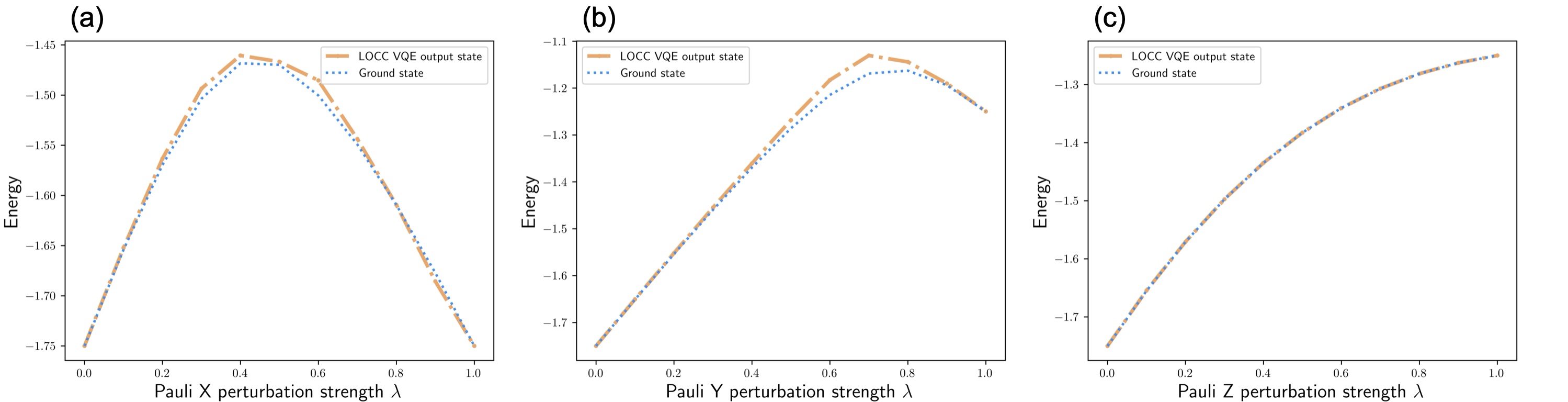}
    \caption{Numerical simulation results of solving the parent Hamiltonian of the four-qubit GHZ state with Pauli $X$ \textbf{(a)}, Y \textbf{(b)}, and Z \textbf{(c)} perturbations. The results are achieved through sample-based LOCC-VQE.}
    \label{fig:GHZML4}
\end{figure}

\begin{figure}[htbp!]
    \centering
    \includegraphics[width=0.4\linewidth]{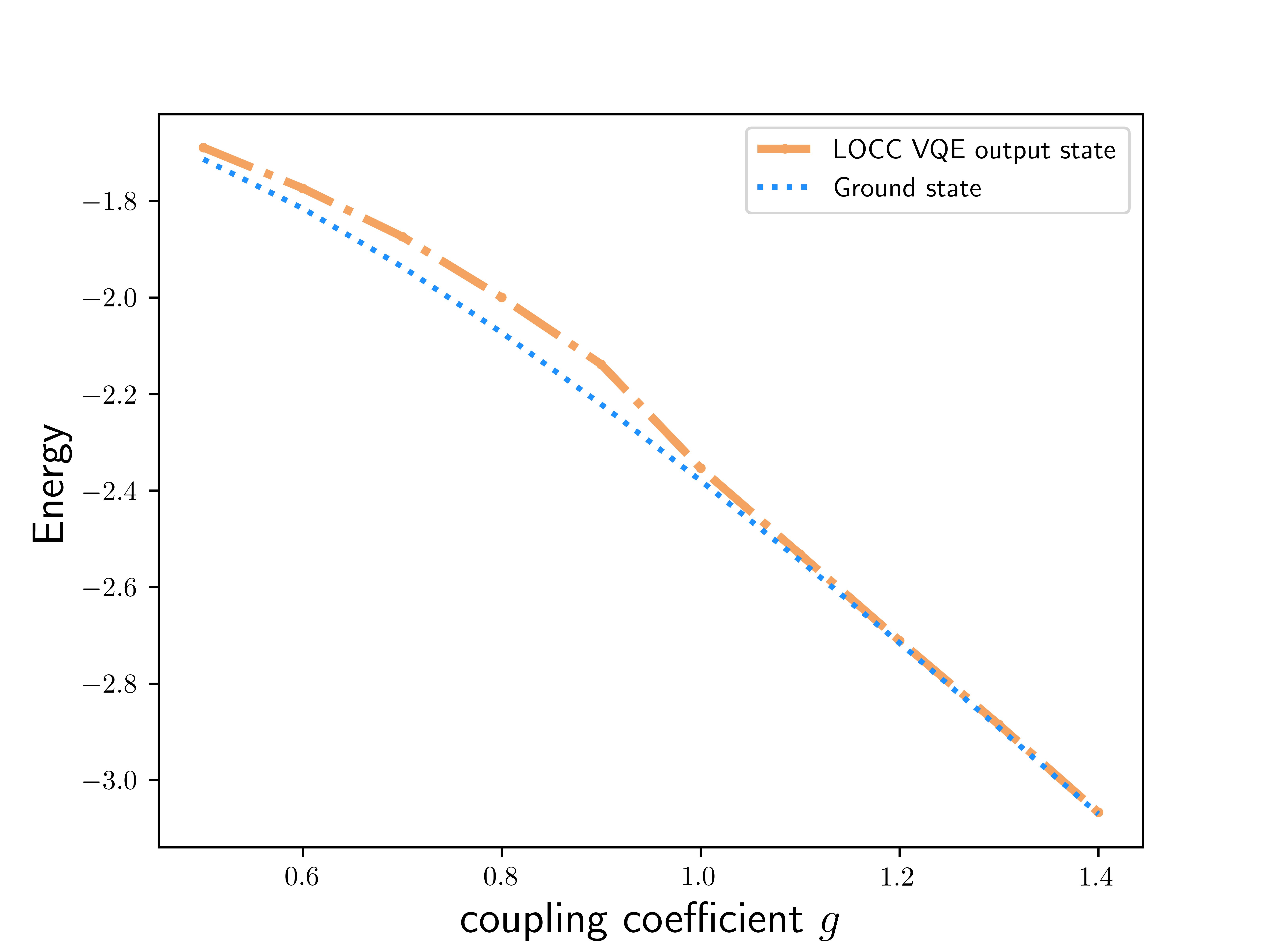}
    \caption{Numerical simulation results of solving transverse-field Ising model with four qubits by LOCC-VQE with finite shots.}
    \label{fig:IsingML4}
\end{figure}

\section{Limitations of unitary circuits in creating long-range entanglement} \label{appendix:QMI_def}

Quantum mutual information (QMI) measures the correlation between subsystems of the quantum state. It is a quantum mechanical analog of Shannon's mutual information. Here we revisit the concept and use it to demonstrate the limitations of unitary circuits in terms of long-range entanglement, quantified by QMI between distant parts.

Consider a quantum system that can be divided into two non-overlapping subsystems $A$ and $B$. The Hilbert space can be written as the tensor product of two sub-spaces corresponding to the division $\mathcal{H}_{AB} = \mathcal{H}_A\otimes\mathcal{H}_B$. For a quantum state $\rho^{AB}$, $\rho^A=\Tr_B(\rho^{AB})\in \mathcal{D}(\mathcal{H}_A)$, and $\rho^B=\Tr_A(\rho^{AB})\in\mathcal{D}(\mathcal{H}_B)$ represent the reduced density matrices of $\rho^{AB}$ each subsystem. Then QMI between subsystems $A$ and $B$ is defined as $I(A,B) = S(\rho^A)+S(\rho^B)-S(\rho^{AB})$, where $S(\cdot)$ is the von Neumann entropy.

QMI between distant parts can be used as a probe of long-range entanglement. Here, we establish this relation by showing the impossibility of a unitary circuit in establishing correlations outside the light cones.

\begin{lemma}
    If two subsystems $A$ and $B$ of the output quantum state of a unitary circuit have non-overlapping information propagation light cones, the quantum mutual information between them is zero.
\end{lemma}

\begin{proof}
Let $\mathcal{L}_A$ be the index set of qubits in the light cone of subsystem $A$, and $\mathcal{L}_B$ be the index set of qubits in the light cone of subsystem $B$. Since subsystems $A$ and $B$ of the output quantum state of this quantum circuit have non-overlapping light cones, we have
\begin{equation}
    \mathcal{L}_A \cap \mathcal{L}_B = \emptyset.
\end{equation}

Denote the $n$-qubit output state of the circuit as $\rho$. Denote the output state tracing out qubits not in the joint light cone of $A$ and $B$ by $\rho^{\mathcal{L}_A,\mathcal{L}_B} = \Tr_{[n]\setminus(\mathcal{L}_A\cup\mathcal{L}_B)}(\rho)$ and tracing out qubits not in the light cone of $A$ and $B$, respectively, by $\rho^{\mathcal{L}_A} = \Tr_{[n]\setminus\mathcal{L}_A}(\rho)$ and $\rho^{\mathcal{L}_B} = \Tr_{[n]\setminus\mathcal{L}_B}(\rho).$ Without overlap between the  light cones, we can write $\rho^{\mathcal{L}_A,\mathcal{L}_B}$ in the product form $\rho^{\mathcal{L}_A,\mathcal{L}_B} = \rho^{\mathcal{L}_A}\otimes\rho^{\mathcal{L}_B}$. Notice that $A\in\mathcal{L}_A$, and $B\in\mathcal{L}_B$, we can express subsystem $A$ and $B$ of the output state as $\rho^{AB} = \rho^A \otimes \rho^B$, where $\rho^A = \Tr_{[n]\setminus A}(\rho)$, and $\rho^B = \Tr_{[n]\setminus B}(\rho)$. Thus, we have
\begin{equation}
I(A,B) = S(\rho^A)+S(\rho^B)-S(\rho^{AB}) = S(\rho^A)+S(\rho^B) - S(\rho^A\otimes\rho^B) = 0.
\end{equation}
\end{proof}

\end{document}